\documentclass{sig-alt-full}
\makeatletter

\usepackage{cite}
\usepackage{graphicx}

\usepackage{subfigure}
\usepackage{url}

\usepackage{array}
\usepackage{mdwmath}
\usepackage{eqparbox}

\usepackage{theorem}
\usepackage{verbatim}
\usepackage{amsmath,amssymb}
\usepackage{multirow}
\usepackage{bbm}
\usepackage{algorithmic}
\usepackage{algorithm}
\usepackage{wrapfig}

\newtheorem{lem}{Lemma}
\newtheorem{prop}{Proposition}
\newtheorem{defn}{Definition}

\newcommand{\bb}[1]{\mathbb{#1}}

\newcommand{\A}[1]{\mathcal{#1}}

\newcommand{\fref}[1]{Fig.~\ref{#1}}
\newcommand{\sref}[1]{Section~\ref{#1}}
\newcommand{\tref}[1]{Table~\ref{#1}}
\newcommand{\pref}[1]{Prop.~\ref{#1}}
\newcommand{\lref}[1]{Lemma~\ref{#1}}

\begin{document}
\conferenceinfo{MobiHoc'08,} {May 26--30, 2008, Hong Kong SAR, China.}
\CopyrightYear{2008}
\crdata{978-1-60558-083-9/08/05}

\title{Analyzing DISH for Multi-Channel MAC Protocols\\in Wireless Networks}

\numberofauthors{2}
\author{
    \alignauthor Tie Luo and Mehul Motani\\
      \affaddr{Electrical \& Computer Engineering}\\
      \affaddr{National University of Singapore}\\
      \email{\{tie,motani\}@nus.edu.sg}
    \alignauthor Vikram Srinivasan\\
      \affaddr{Bell Labs Research}\\
      \affaddr{India}\\
      \email{vikramsr@alcatel-lucent.com}
}

\maketitle

\begin{abstract}
For long, node cooperation has been exploited as a {\em data relaying} mechanism. However, the wireless channel allows for much richer interaction between nodes. One such scenario is in a multi-channel environment, where transmitter-receiver pairs may make incorrect decisions (e.g., in selecting channels) but idle neighbors could help by sharing information to prevent undesirable consequences (e.g., data collisions). This represents a {\em Distributed Information SHaring} (DISH) mechanism for cooperation and suggests new ways of designing cooperative protocols.
However, what is lacking is a theoretical understanding of this new notion of cooperation. In this paper, we view cooperation as a network resource and evaluate the {\em availability of cooperation} via a metric, $p_{co}$, the probability of obtaining cooperation.
First, we analytically evaluate $p_{co}$ in the context of multi-channel multi-hop wireless networks. Second, we verify our analysis via simulations and the results show that our analysis accurately characterizes the behavior of $p_{co}$ as a function of underlying network parameters. This step also yields important insights into DISH with respect to network dynamics. Third, we investigate the correlation between $p_{co}$ and network performance in terms of collision rate, packet delay, and throughput. The results indicate a near-linear relationship, which may significantly simplify performance analysis for cooperative networks and suggests that $p_{co}$ be used as an appropriate performance indicator itself. Throughout this work, we utilize, as appropriate, three different DISH contexts --- model-based DISH, ideal DISH, and real DISH --- to explore $p_{co}$.
\end{abstract}

\vspace{-2mm}

\vspace{-3mm}

\vspace{-3mm}

\section{Introduction} \label{sec:intro}

Cooperative diversity is not a new concept in wireless communications.  Key ideas and results in cooperative communications can be traced back to the 1970s to van der Meulen \cite{van3term} and Cover \& El Gamal~\cite{cover79}, whose works have spurred numerous studies on this topic from an information-theoretic perspective (e.g., \cite{coopcap06,laneman03,coop03toc1,coop03toc2}) or a protocol-design perspective (e.g., \cite{rdcf05infocom,cmac05globecom,cdmac07icc,comac07jsac}).
To date, cooperation has been intensively studied in various contexts. However, to the best of our knowledge, it has always been used as a {\em data relaying} mechanism where intermediate nodes help relay packets from a transmitter to a receiver. In fact, the wireless channel allows for much richer interaction among nodes. Consider a scenario where orthogonal frequency channels are available. A node wishes to select a conflict-free channel to transmit data, but may often fail to achieve this due to lack of sufficient information about channel usage. In this case, other nodes in the neighborhood may possess the information in need and thus could help by sharing this information.

This shows that cooperation can be used as a {\em Distributed Information SHaring} (DISH) mechanism, in addition to mere data relaying. In \cite{tie06cam,tie07mobicom}, we proposed a multi-channel MAC protocol based on this idea, where performance enhancement was demonstrated via simulations. In this paper, we develop a theoretical treatment of this new notion of cooperation, in particular, the {\em availability of cooperation}.  The benefit of DISH is that it can remove the need of using multiple transceivers \cite{dca00,nas99,jain01,mup04,mah06} and time synchronization \cite{chen03,so04,tmmac07,chma00,chat00,ssch04} in designing multi-channel MAC protocols.  This motivates us to understand DISH from a theoretical perspective.

In this paper, we define a metric $p_{co}$ which characterizes the availability of cooperation as {\em the probability of obtaining cooperation} (see Def.~\ref{def:pco} for a more precise definition). We analytically evaluate this metric in multi-channel multi-hop wireless networks with randomly distributed nodes, and verify the analysis via simulations. We also carry out a detailed investigation of $p_{co}$ with three different contexts of DISH: model-based DISH, ideal DISH, and real DISH, in order to obtain meaningful findings.

\subsection{Summary of Contributions}

Our aim in this paper is to understand DISH and the availability of cooperation ($p_{co}$) from an analytical perspective. More specifically,
we provide an analysis which accurately characterizes the availability of cooperation as a function of the underlying network parameters. This analysis reveals {\em what} underlying factors and {\em how} these factors affect cooperation, and can provide guidelines to provisioning the network to increase performance.

Throughput analysis for multi-hop networks is difficult (and still an open problem in general), and it gets even more complicated in a multi-channel context with DISH.  Our approach in this paper is to first look at $p_{co}$ and then correlate it with network performance.  The results indicate that there is a simple relationship between $p_{co}$ and several performance metrics.

The specific findings of this study are:
\begin{enumerate}
  \item The availability of cooperation is high ($p_{co}>0.7$) in typical cases, which suggests that DISH is feasible to use in multi-channel MAC protocols.
  \item The performance degradation due to an increase in node density can be alleviated due to the simultaneously increased availability of cooperation.
  \item The metric $p_{co}$ will increase for larger packet sizes for a given {\em bit} arrival rate, but will decrease for larger packet sizes for a given {\em packet} arrival rate.
  \item Node density and traffic load have opposite effects on $p_{co}$ but node density is the dominating factor. This implies an improved scalability for DISH networks as $p_{co}$ increases with node density.
  \item $p_{co}$ is strongly correlated to network performance and has a near-linear dependence with metrics such as throughput and delay. This may significantly simplify performance analysis for cooperative networks, and suggests that $p_{co}$ be used as an appropriate performance indicator itself.
\end{enumerate}

\section{Related Work}\label{sec:rel}

\ifdefined\JNL
Multi-channel MAC protocols for ad hoc networks can be categorized into multi-radio schemes and single-radio schemes. The first category \cite{wu00,nas99,jain01,mup04,mah06} uses one more radio to monitor channel usage when the other one is engaged in data communication. The second category mainly uses time synchronization to avoid MCC problems, by requiring nodes to negotiate channels in well-known time slots \cite{chen03,so04,tmmac07} or hop among channels according to well-known or predictable sequences \cite{chma00,chat00,ssch04}. The multi-radio schemes are cost inefficient and the synchronous schemes have large overhead and limited scalability.

The core idea of CAM-MAC\cite{tie06cam,tie07tmcsub} is using cooperation to solve MCC problems. Essentially, each node uses its {\em idle neighbors'} radios to gather channel usage information when itself cannot, and then retrieve the information back when needed. This simple idea of distributed information storage and retrieval led to a single-radio and fully asynchronous solution.
\fi

There are three other studies most related to this work. One is CAM-MAC~\cite{tie06cam} which
uses cooperation in the new way that we call DISH in this paper. It is a cooperative multi-channel MAC protocol requiring only a single transceiver and no synchronization. In this protocol, there is a control channel for transmitter-receiver pairs to perform handshakes in order to reserve data channels, while nodes in the neighborhood may send cooperative messages to invalidate the handshake if the selected channel or receiver is busy.
As it is the only work we are aware of that explicitly uses DISH, our system model will use a protocol framework by abstracting the work in \cite{tie06cam}.

The second work, CoopMAC~\cite{comac07jsac}, is also a cooperative MAC protocol which exploits data relaying as many other protocols do, such as \cite{rdcf05infocom,cmac05globecom,cdmac07icc}.
A protocol analysis is provided in the paper and it requires computing the probability that a {\em relay node} is available. This probability is different from $p_{co}$ in that it is determined by the {\em static locations} of nodes, i.e., whether a node exists in a specific region.
The probability is computed via geometric analysis (nodes are assumed to be uniformly distributed). On the other hand, $p_{co}$ is determined not only by static node locations, but also by {\em dynamic node behavior}, e.g, a node must have acquired the specific information at a specific moment. The second main difference between the protocol analysis of CoopMAC and our work is that the problem context of CoopMAC is a {\em wireless LAN} with {\em a single channel}, whereas this paper assumes a multi-hop network with multiple channels.

The last work is by Han et al. \cite{han05} who considered a multi-channel MAC protocol adopting ALOHA on the control channel to reserve data channels.  A queueing-theoretic approach was taken to calculate throughput for the protocol. However, there are some noteworthy limitations. First, only a {\em single-hop} scenario was considered. Second, each node was assumed to be able to communicate on the control channel and a data channel simultaneously. This essentially requires {\em two transceivers} per node, and consequently leads to collision-free data channels, which oversimplifies the problem. Third, a unique virtual queue was assumed to store the packets arriving at {\em all} nodes for the ease of {\em centralized transmission scheduling}, and the precise status of the queue was assumed to be known to the entire network. This assumption is impractical and eventually results in a throughput upper bound. Fourth, the access to the control channel adopts the ALOHA algorithm, rather than the more practical and sophisticated mechanism of CSMA/CA.

\section{System Model}\label{sec:model}

We consider a static and connected ad hoc network in which each node is equipped with a single half-duplex transceiver that can dynamically switch between a set of orthogonal frequency channels but can only use one at a time. One channel is designated as a control channel and the others are designated as data channels. Nodes are placed in a plane area according to a two-dimensional Poisson point process.

We consider a class of multi-channel MAC protocols with their common framework described below. A transmitter-receiver pair uses a \texttt{McRTS}/\texttt{McCTS} handshake on the control channel to set up communication (like 802.11 RTS/CTS) for their subsequent \texttt{DATA}/\texttt{ACK} handshake on a data channel. To elaborate, a transmitter sends a \texttt{McRTS} on the control channel using CSMA/CA, i.e., it sends \texttt{McRTS} after sensing the control channel to be idle for a random period (addressed below) of time. The intended receiver, after successfully receiving \texttt{McRTS}, will send a \texttt{McCTS} and then switch to a data channel (the \texttt{McRTS} informs the receiver of the data channel).
After successfully receiving the \texttt{McCTS}, the transmitter will also switch to its selected data channel, and otherwise it will backoff on the control channel for a random period (addressed below) of time. Hence it is possible that only the receiver switches to the data channel. After switching to a data channel, the transmitter will send a \texttt{DATA} and the receiver will respond with a \texttt{ACK} upon successful reception. Then both of them switch back to the control channel.

In the above we have mentioned two random periods of time. Our model does not specify them but assumes that these periods are designed such that idle intervals on the control channel are well randomized. Specifically, when a node is on the control channel, it sends control messages (an aggregated stream of \texttt{McRTS} and \texttt{McCTS}) according to a Poisson process.

Note that we use \texttt{McRTS}, \texttt{McCTS}, \texttt{DATA} and \texttt{ACK} to refer to different packets (frames) without assuming specific frame formats. Since, logically, they must make a protocol functional, we assume that \texttt{McRTS} carries channel usage information (e.g., ``who will use which channel for how long'') and, for simplicity, \texttt{McCTS} is the same as \texttt{McRTS}.

We assume that, after switching to a data channel, a node will stay on that channel for a period of $T_d$, where $T_d$ is the duration of a successful data channel handshake.
We ignore channel switching delay as it will not fundamentally change our results if it is negligible compared to $T_d$ (the delay is $80\mu s$\cite{ssch04} while $T_d$ is more than $6ms$ for a 1.5KB data packet on a 2Mb/s channel). We also ignore SIFS and propagation delay for the same reason, provided that they are smaller than the transmission time of a control message.

We assume a uniform traffic pattern --- all nodes have the equal data packet arrival rate, and for each data packet to send, a node chooses a receiver equally likely among its neighbors.
We also assume a stable network --- all data packets can be delivered to destinations within finite delay. In addition, packet reception fails if and only if packets collide with each other (i.e., no capture effect), transmission and interference ranges are equal, and the probability that {\em neighboring} nodes simultaneously {\em start} sending control messages is zero (no time synchronization).

We do not assume a specific channel selection strategy; how a node selects data channels will affect how often conflicting channels are selected, but will not affect $p_{co}$. This is because, intuitively, we only care about the availability of cooperation ($p_{co}$) when a multi-channel coordination problem (a precise definition is given in Def.~\ref{def:mccp}), which includes channel conflicting problem, {\em has been} created.

We do not assume a concrete DISH mechanism, i.e., nodes do not physically react upon a multi-channel coordination problem, because analyzing the availability of cooperation does not require the use of this resource. In fact, assuming one of the (numerous possible) DISH mechanisms will lose generality. Nevertheless, we will show in \sref{sec:simu} that, when an ideal or a real DISH mechanism is used, the results do not fundamentally change. This could be an overall effect from contradicting factors which will be explained therein.

The following lists all parameters that are assumed known:
\begin{itemize}
\item $n$: node density. In a multi-hop network, it is the average number of nodes per $R^2$ where $R$ is the transmission range. In a single-hop network, it is the total number of nodes.
\item $\lambda$: the average data packet arrival rate at each node, including retransmissions.
\item $T_d$: the duration of a data channel handshake.
\item $b$: the transmission time of a control message. $b\ll T_d$.
\end{itemize}

\renewcommand{\labelenumi}{(\alph{enumi})}

\section{Analysis} \label{sec:analysis}

\subsection{Problem Formulation and Analysis Outline}\label{sec:formulate}

We first formally define $p_{co}$, which depends on two concepts called the {\em MCC problem} and the {\em cooperative node}.

\begin{defn}[MCC Problem]\label{def:mccp}
A multi-channel coordination (MCC) problem is either a {\em channel conflict} problem or a {\em deaf terminal} problem. A channel conflict problem is created when a node, say $y$, selects a channel to use (transmit or receive packets) but the channel is already in use by a neighboring node, say $x$. A deaf terminal problem is created when a node, say $y$, initiates communication to another node, say $x$, that is however on a different channel. In either case, we say that an MCC problem is created by $x$ and $y$.
\end{defn}

In a protocol that transmits {\tt DATA} without requiring {\tt ACK}, a channel conflict problem does not necessarily indicate an impending data collision. We do not consider such a protocol.

\begin{defn}[Cooperative Node]\label{def:coopnode}
A node that identifies an MCC problem created by two other nodes, say $x$ and $y$, is called a cooperative node with respect to $x$ and $y$.
\end{defn}

See \fref{fig:coopnode} for a visualization based on our system model.

\begin{figure}[tbp]
\centering
\includegraphics[width=0.8\linewidth]{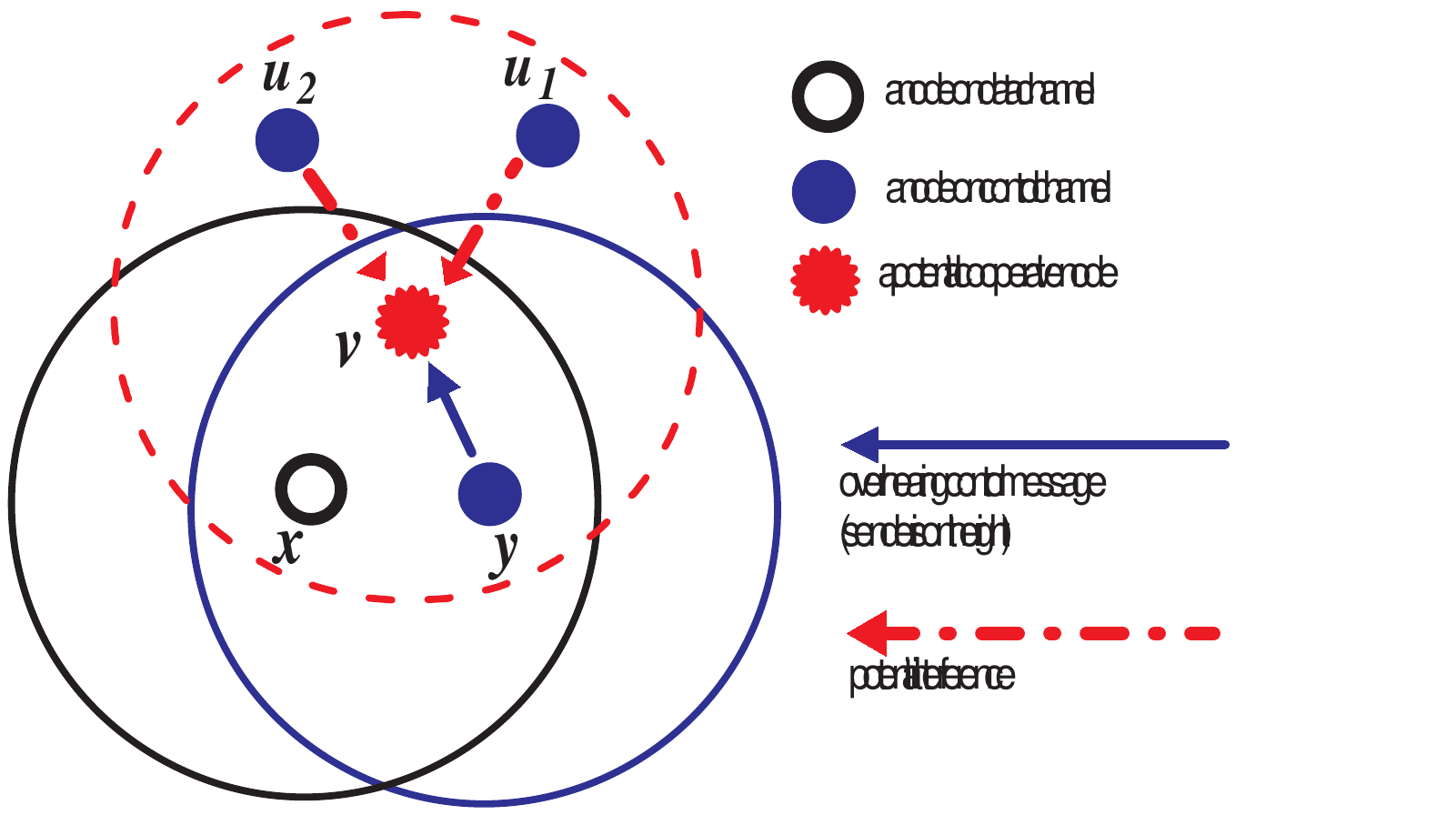}
\caption{Illustration of an MCC problem and a cooperative node. Node $x$ is performing a data channel channel handshake on $CH_x$, and $y$ has just sent a control message during a control channel handshake. If this control message is an {\tt McRTS} addressed to $x$, then a {\em deaf terminal} problem is created. If this control message indicates that $y$ selects $CH_x$ (recall that a control message carries channel usage information), a {\em channel conflict} problem is created. In either case, if a third node $v$ identifies this problem (by overhearing $x$'s and $y$'s control messages successively), it is a cooperative node.}
\label{fig:coopnode}
\end{figure}

\begin{defn}[$p_{co}$]\label{def:pco}
$p_{co}$ is the probability for two arbitrary nodes that create an MCC problem to obtain cooperation, i.e., there is at least one cooperative node with respect to these two nodes.
\end{defn}

Note that, if there are multiple cooperative nodes and a DISH mechanism allows those nodes to send cooperative messages concurrently, then a collision results. However, this collision {\em still indicates an MCC problem} and thus cooperation is still deemed obtained. CAM-MAC\cite{tie06cam} also implements this.

We distinguish the receiving of control messages. A transmitter receiving \texttt{McCTS} from its intended receiver is referred to as {\em intentional receiving}, and the other cases of receiving are referred to as {\em overhearing}, i.e., any node receiving \texttt{McRTS} (hence an intended receiver may also be a cooperative node) or any node other than the intended transmitter receiving \texttt{McCTS}.

Our notation is listed in \tref{tab:notation}. Overall, we will determine $p_{co}$ by following the order of $p_{co}^{xy}(v)\rightarrow p_{co}^{xy}\rightarrow p_{co}$.
\begin{table*}[tbp]
\caption{Notation}\label{tab:notation}
\centering
\begin{tabular} {| c | l | p{5.6in} | }
\hline
\multirow{5}{*}{\rotatebox{90}{\mbox{Probabilities}}}
& $p_{co}^{xy}$ & the probability that at least one cooperative node with respect to $x$ and $y$ exists \\
& $p_{co}^{xy}(v)$ & the probability that node $v$ is a cooperative node with respect to $x$ and $y$\\
& $p_{ctrl}$ & the probability that a node is on the control channel at an arbitrary point in time \\
& $p_{succ}$ & the probability that a control channel handshake (initiated by a \texttt{McRTS}) is successful \\
& $p_{oh}$ & the probability that an arbitrary node successfully overhears a control message \\
\hline
\multirow{6}{*}{\rotatebox{90}{\mbox{Events}}}
& $\A C_v(t)$ & node $v$ is on the control channel at time $t$ \\
& $\A O (v\leftarrow i)$ & node $v$ successfully overhears node $i$'s control message, given that $i$ sends the message \\
& $\A S_v(t_1,t_2)$ & node $v$ is silent (not transmitting) on the control channel during interval $[t_1, t_2]$ \\
& $\A I_v(t_1,t_2)$ & node $v$ does not introduce interference to the control channel during interval $[t_1, t_2]$, i.e., it is on a data channel or is silent on the control channel. \\
& $\Omega_u(t_1,t_2)$ & node $u$, which is on a data channel at $t_1$, switches to the control channel in $[t_1,t_2]$ \\
\hline
\multirow{5}{*}{\rotatebox{90}{\mbox{Others}}}
& $\A N_i, \A N_{ij}, \A N_{v\backslash i}$ & $\A N_i$ is the set of node $i$'s neighbors,
$\A N_{ij}=\A N_i \cap \A N_j$, $\A N_{v\backslash i}=\A N_v\backslash \A N_i\backslash \{i\}$ ($v$'s but not $i$'s neighbors) \\
& $K_{ij}, K_{v\backslash i}$ & $K_{ij}=|\A N_{ij}|$, $K_{v\backslash i}=|\A N_{v\backslash i}|$ \\
& $s_i$ & the time when node $i$ starts to send a control message\\
& $\lambda_c, \lambda_{rts}, \lambda_{cts}$ & the average rates of a node sending control messages, \texttt{McRTS}, and \texttt{McCTS}, respectively, {\em when it is on the control channel}. Clearly, $\lambda_c = \lambda_{rts}+\lambda_{cts}$. \\
\hline
\end{tabular}
\end{table*}

Consider $p_{co}^{xy}(v)$ first. \fref{fig:coopnode} illustrates that node $v$ is cooperative if and only if it successfully overhears $x$'s and $y$'s control messages successively. Hence $\forall v\in \A N_{xy}$,
\begin{align}\label{eq:pcoxyv-def}
p_{co}^{xy}(v) &= \Pr[\A O(v\leftarrow x),\A O(v\leftarrow y)]\notag\\
&= \Pr[\A O (v\leftarrow x)]\cdot \Pr[\A O (v\leftarrow y)|\A O (v\leftarrow x)].
\end{align}

Consider $\A O (v\leftarrow i)$. For $v$ to successfully overhear $i$'s control message which is being sent during interval $[s_i,s_i+b]$, $v$ must be silent on the control channel and not be interfered, i.e.,
\begin{align}\label{eq:poh-def}
\Pr[\A O (v\leftarrow i)] =& \Pr[\A S_v(s_i, s_i+b),
    \bigcap_{u \in \A N_v\backslash \{i\}} \A I_u(s_i,s_i+b)],\notag\\
    &\forall v \in \A N_i.
\end{align}

Now we outline our analysis as below.
\begin{itemize}
\item \sref{sec:basicprop}: solves \eqref{eq:poh-def}, with $p_{ctrl}$ and $\lambda_c$ introduced.
\item \sref{sec:balance}: solves $p_{ctrl}$ and $\lambda_c$.
\item \sref{sec:revisit}: solves \eqref{eq:pcoxyv-def} and then the target metric $p_{co}$.
\item \sref{sec:sghop}: special-case study in single-hop networks.
\end{itemize}

\subsection{Solving Equation~\eqref{eq:poh-def}} \label{sec:basicprop}

\begin{prop} \label{prop:switch}
If node $u$ is on a data channel at $t_1$, then the probability that $u$ does not introduce interference to the control channel during $[t_1,t_2]$, where $t_2-t_1=\Delta t<T_d$, is given by
\begin{align*}
 \Pr[\A I_u(t_1,t_2)| \overline{\A C_u(t_1)}]
=1-\frac{\Delta t}{T_d} + \frac{1- e^{-\lambda_c \Delta t}}{\lambda_c T_d}.
\end{align*}
\end{prop}

\begin{proof}
By the total probability theorem,
\begin{align*}
l.h.s.=&\Pr[\Omega_u(t_1,t_2)]\times\Pr[\A I_u(t_1,t_2) |\Omega_u(t_1,t_2)]\notag\\
 &+\Pr[\overline{\Omega_u(t_1,t_2)}]\times 1.
\end{align*}

Let $t_{sw}$ be the time when node $u$ switches to the control channel (see \fref{fig:switch}). It is uniformly distributed in $[t_1, t_1+T_d]$ because the time when $u$ started its data channel handshake is unknown, and hence
\begin{align}\label{eq:switch}
 \Pr[\Omega_u(t_1,t_2)] = \frac{\Delta t}{T_d}.
\end{align}

\begin{figure}[tbp]
\centering
  \includegraphics[trim=0 0 0 3mm,clip,width=0.75\linewidth]{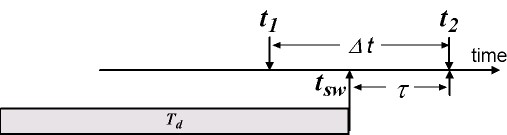}
  \caption{A node switches to the control channel after data channel handshaking.}
  \label{fig:switch}
\vspace{-3mm}\end{figure}

Since control channel traffic is Poisson with rate $\lambda_c$,
\begin{align*}
\Pr[\A I_u(t_1, t_2) | \Omega_u(t_1,t_2)]
  =\Pr[\A S_u(t_{sw}, t_2) | \Omega_u(t_1,t_2)]
  =\bb E [e^{-\lambda_c \tau}]
\end{align*}
where $\tau=t_2 - t_{sw}$ is uniformly distributed in $[0,\Delta t]$ by the same argument leading to \eqref{eq:switch}. Hence
\begin{align*}
\mathbb{E} [e^{-\lambda_c \tau}]
    = \int_0^{\Delta t} e^{-\lambda_c \tau_0} \frac{1}{\Delta t} d\tau_0
    = \frac{1- e^{- \lambda_c \Delta t}}{\lambda_c \Delta t},
\end{align*}
and then by substitution the proposition is proven.
\end{proof}

\begin{prop}\label{prop:notintfoh}
If node $v$ is overhearing a control message from node $i$ during $[s_i, s_i+b]$, then the probability that a node $u \in \A N_v$ does not interfere with $v$ is given by
\begin{align*}
\Pr[\A I_u(s_i, s_i+b)] = \left\{
  \begin{array}{ll}
    1, & u \in \A N_{vi}; \\
    p_{ni\text{-}oh}, & u \in \A N_{v\backslash i}.
  \end{array}
\right.
\end{align*}
where
\[ p_{ni\text{-}oh} = \ p_{ctrl} \cdot e^{-2 \lambda_c b} +(1-p_{ctrl})
    \cdot (1-\frac{2b}{T_d} + \frac{1- e^{-2 \lambda_c b}}{\lambda_c T_d}). \]
\end{prop}

\begin{prop}\label{prop:notintfcts}
If node $i$ (transmitter) is intentionally receiving \texttt{McCTS} from node $j$ (receiver) during $[s_j, s_j+b]$, then the probability that a node $u \in \A N_i$ does not interfere with $i$ is given by
\begin{align*}
\Pr[\A I_u(s_j, s_j+b)] = \left\{
  \begin{array}{ll}
    1, & u \in \A N_{ij}; \\
    p_{ni\text{-}cts}, & u \in \A N_{i\backslash j}.
  \end{array}
\right.
\end{align*}
where $p_{ni\text{-}cts} =$
\[ (1-p_{ctrl}) [1-\frac{b}{T_d}
       (1+\frac{b}{T_d} - \frac{1- e^{-\lambda_c b}}{\lambda_c T_d} - e^{-\lambda_c b})]
    + p_{ctrl}. \]
\end{prop}

Then we can solve \eqref{eq:poh-def}, based on \pref{prop:notintfoh}, to be
\begin{align}\label{eq:poh-vi}
\Pr[\A O (v\leftarrow i)] \approx p_{ctrl} \; p_{ni\text{-}oh}^{K_{v\backslash i}}.
\end{align}

See appendix for the proof of \pref{prop:notintfoh}, \pref{prop:notintfcts} and Eq.~\eqref{eq:poh-vi}.

\subsection{Solving $p_{ctrl}$ and $\lambda_c$}\label{sec:balance}

For $p_{ctrl}$, consider a sufficiently long period $T_0$. On the one hand, the number of arrival data packets at each node is $\lambda T_0$. On the other hand, each node spends total time of $(1-p_{ctrl})T_0$ on data channels, a factor $\eta$ of which is used for sending arrival data packets. Since the network is stable (incoming traffic is equal to outgoing traffic), we establish a balanced equation:
\begin{align*}
    \lambda T_0 T_d = \eta\; (1-p_{ctrl}) T_0.
\end{align*}
To determine $\eta$, noticing that a node switches to data channels either as a transmitter (with an average rate of $\lambda$) or as a receiver (with an average rate of $\lambda_{cts}$), we have
$\eta =\lambda/(\lambda+\lambda_{cts})$. Substituting this into the above yields
\begin{align}\label{eq:pctrl}
p_{ctrl} = 1 - (\lambda+\lambda_{cts}) T_d.
\end{align}

For $\lambda_c$ (together with $\lambda_{cts}$), we need two lemmas.
\begin{lem} \label{lem:epk}
For a Poisson random variable $K$ with mean value $\overline{K}$, and $0<p<1$,
\begin{equation*}
    \mathbb{E} [p^K] = e^{-(1-p)\overline{K}}.
\end{equation*}
\end{lem}

\begin{proof}
\begin{align*}
    \mathbb{E} [p^K] = \sum_0^\infty p^k \Pr(K=k)
        = e^{-\overline{K}} \sum_0^\infty \frac{(p\overline{K})^k}{k!}
        = e^{-(1-p)\overline{K}}.
\end{align*}
\end{proof}

\vspace{-5mm}
\begin{lem} \label{lem:avgarea}
For three random distributed nodes $v$, $i$ and $j$,
\begin{enumerate}
\item $\mathbb{E}[K_{v\backslash i}|v\in \A N_i] \approx 1.30 n$.
\item $\mathbb{E}[K_{v\backslash i}|v\in \A N_{ij}] \approx 1.19 n$.
\item $\bb{E}[K_{ij}] \approx 1.84 n$.
\end{enumerate}
\end{lem}

\begin{proof}
Let $A_s(\gamma)$ be the intersection area of two circles with a distance of $\gamma$ between their centers, and $\gamma<R$, where $R$ is the circles' radius. It can be derived from \cite{mathwd-cirseg} that
\begin{align*}
    A_s(\gamma) = 2 R^2 \arccos\frac{\gamma}{2R} -
        \gamma\sqrt{R^2-\frac{\gamma^2}{4}}.
\end{align*}
Let $A_c(\gamma)$ be the complementary area of $A_s(\gamma)$, i.e., $A_c(\gamma) = \pi R^2 - A_s(\gamma)$,
and let $A_{ij}$ and $A_{v\backslash i}$ be the areas where $\A N_{ij}$ and $\A N_{v\backslash i}$ are located, respectively.

(a) See \fref{fig:distpdf-all}. Letting $\gamma=||vi||$ where $v\in \A N_i$, and $f(r)$ be its probability density function (pdf), we have $f(r) dr = 2 \pi r dr/(\pi R^2)$, which gives $f(r) = 2 r/R^2$. Thus
\begin{align*}
    \mathbb{E}[A_{v\backslash i}|v\in \A N_i]
    = \int_0^R A_c(r) f(r) d r \notag
    \approx 1.30 R^2,
\end{align*}
and hence $\mathbb{E}[K_{v\backslash i}|v\in \A N_i] \approx n\cdot 1.30 R^2/R^2 = 1.30 n$.

\begin{figure}[tbp]
\centering
\subfigure[$v\in \A N_i$.]
    {\includegraphics[width=0.32\linewidth]{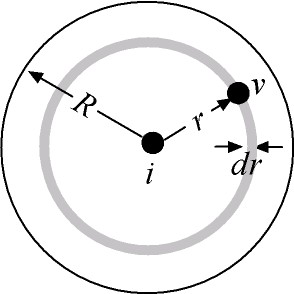}
    \label{fig:distpdf-all}} \hfil
\subfigure[$v\in\A N_{i\backslash j}$.]
    {\includegraphics[width=0.55\linewidth]{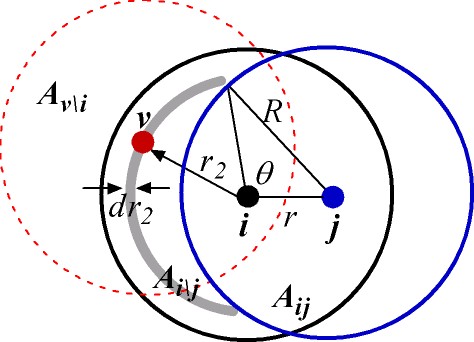}
    \label{fig:distpdf-moon}}
\caption{Deriving the pdf of distance $||vi||$.}
\label{fig:distpdf}
\vspace{-5mm}\end{figure}

(b) Let $\gamma_1=||vi||$ where $v\in \A N_{ij}$, and $f_1(r_1)$ be its pdf.
To solve $f_1(r_1)$, we consider $v\in \A N_{i \backslash j}$ instead (see \fref{fig:distpdf-moon}):
\begin{align}\label{eq:condexp}
&\because\ \mathbb{E}[A_{v\backslash i}|v\in \A N_i]
  = p_1 \cdot \mathbb{E}[A_{v\backslash i}|v\in \A N_{ij}]\notag\\
&\hspace{3.1cm}+(1-p_1) \cdot \mathbb{E}[A_{v\backslash i} | v\in \A N_{i \backslash j}]\notag\\
&\hspace{5mm}\text{where } p_1\triangleq \Pr[v\in \A N_{ij} | v\in\A N_i] = \frac{A_s(r)}{\pi R^2},\notag\\
&\therefore\ \mathbb{E}[A_{v\backslash i} | v\in \A N_{ij}]
    = p_1^{-1}\mathbb{E}[A_{v\backslash i}|v\in \A N_i]\notag\\
&\hspace{3.1cm}- (p_1^{-1}-1)\cdot \mathbb{E}[A_{v\backslash i} | v\in \A N_{i \backslash j}].
\end{align}
To determine $\mathbb{E}[A_{v\backslash i} | v\in \A N_{i\backslash j}]$, let $\gamma_2=||vi||$ where $v\in \A N_{i\backslash j}$, and $f_2(r_2)$ be its pdf. It is determined by
\begin{align*}
f_2(r_2) dr_2 = \frac{2(\pi - \theta)r_2 dr_2}{A_{i\backslash j}} \text{ and }
\cos\theta = \frac{r_2^2 + r^2 - R^2}{2r_2 r}.
\end{align*}
Therefore
\begin{align*}
 &\mathbb{E}[A_{v\backslash i} | v\in \A N_{i\backslash j}]
   = \int_{R-r}^R A_c(r_2) f_2(r_2) dr_2\\
=&\int_{R-r}^R \frac{2 r_2 A_c(r_2)}{A_c(r)}
    (\pi - \arccos \frac{r_2^2 + r^2 - R^2}{2r_2 r} ) dr_2.
\end{align*}
Substituting this and $\mathbb{E}[A_{v\backslash i}|v\in \A N_i]\approx 1.30 R^2$ (by case (a)) into \eqref{eq:condexp} solves $\bb{E}[A_{v\backslash i} | v\in \A N_{ij}]$, which we denote by $M(r)$.
Then we have
\[ \mathbb{E}[K_{v\backslash i}|v\in \A N_{ij}] = \frac{n}{R^2} \int_0^R M(r)f(r)dr \approx 1.19 n. \]

(c) Proven by noticing $A_{ij}$ is complementary to the area corresponding to case (a).
\end{proof}

These lemmas enable us to prove (see appendix) that
\begin{align}\label{eq:poh_psucc}
p_{oh} \approx p_{ctrl}\; \exp[-1.30 n (1-p_{ni\text{-}oh})],\notag\\
p_{succ} \approx p_{oh} \exp[-1.30 n(1-p_{ni\text{-}cts})].
\end{align}

Now we solve $\lambda_c$ (together with $\lambda_{cts}$). From the perspective of a transmitter, the average number of successful control channel handshakes that it initiates per second is $p_{ctrl} \lambda_{rts} p_{succ}$. Since each successful control channel handshake leads to transmitting one data packet, we have $p_{ctrl} \lambda_{rts} p_{succ} = \lambda$.

From the perspective of a receiver, it sends a \texttt{McCTS} when it successfully receives (overhears) a \texttt{McRTS} addressed to it,
and hence $\lambda_{cts} = \lambda_{rts}\, p_{oh}$.  Then combining these with $\lambda_c = \lambda_{rts}+\lambda_{cts}$ yields
\begin{align}\label{eq:lmdc}
    \lambda_c = \frac{\lambda (1+p_{oh})}{p_{ctrl}\, p_{succ}} \text{ and }
    \lambda_{cts} = \frac{\lambda\ p_{oh}}{p_{ctrl}\, p_{succ}}
\end{align}
where $p_{oh}$ and $p_{succ}$ are given in \eqref{eq:poh_psucc}.

\subsection{Solving Equation \eqref{eq:pcoxyv-def} and Target Metric $p_{co}$}\label{sec:revisit}

Based on the proof of \eqref{eq:poh-vi}, it can be derived that
\begin{align}\label{eq:poh-vy}
\Pr[\A O (v\leftarrow y)|\A O (v\leftarrow x)] \approx p_{ctrl}^{\star} \; p_{ni\text{-}oh}^{K_{v\backslash y}},
\end{align}
\[ \text{where \ \ } p_{ctrl}^{\star} \triangleq \Pr[\A C_v(s_y)|\A O (v\leftarrow x)]. \]
Note that $p_{ctrl}^{\star} \neq p_{ctrl}$, because $s_y$ is not an {\em arbitrary} time for $v$ due to the effect form $\A O (v\leftarrow x)$. The reason is that $\A O (v\leftarrow x)$ implies $\A C_v(s_x)$, and thus for $\A C_v(s_y)$ to happen, $v$ must stay {\em continuously} on the control channel during $[s_x,s_y]$ (otherwise, a switching will lead to $v$ staying on the data channel for $T_d$, but $s_x+T_d>s_y$ since $x$'s data communication is still ongoing at $s_y$, and hence $\A C_v(s_y)$ can never happen).

It can be proven (see appendix for the proof) that
\begin{align}\label{eq:pctrlstar}
p_{ctrl}^{\star} = \frac{(w\lambda_c - \frac{1-w}{T_d})\; g(\lambda_c+\lambda_w) +
    \frac{1-w}{T_d}\; g(\lambda_w)} {1-w+ (w\lambda_c - \frac{1-w}{T_d})\; g(\lambda_c)}
\end{align}
where
\[ g(x)=\frac{1- e^{-x T_d}}{x},\;\; w = \frac{p_{ctrl}-p_{oh}}{1-p_{oh}}, \text{ and}\]
\[ \lambda_w=\lambda_{rts} p_{succ}+\lambda_{cts}. \]

Combining \eqref{eq:poh-vi} and \eqref{eq:poh-vy} reduces \eqref{eq:pcoxyv} to
\begin{align}\label{eq:pcoxyv}
p_{co}^{xy}(v) \approx p_{ctrl} \; p_{ctrl}^{\star}
  \; p_{ni\text{-}oh}^{K_{v\backslash x}+K_{v\backslash y}},\; \forall  v\in\A N_{xy}.
\end{align}

Let $p_{co}^{xy}(\star)$ be the average of $p_{co}^{xy}(v)$ over all $v\in\A N_{xy}$, i.e., $p_{co}^{xy}(\star)$ is the probability that an arbitrary node in $\A N_{xy}$ is cooperative with respect to $x$ and $y$, Using \lref{lem:avgarea}-(b),
\begin{align}\label{eq:pcostar}
p_{co}^{xy}(\star) \approx p_{ctrl} \; p_{ctrl}^{\star} \; \exp[-2.38 n (1-p_{ni\text{-}oh})].
\end{align}

By the definition of $p_{co}^{xy}$ in \tref{tab:notation},
\begin{align}\label{eq:pcoxy-approx}
    p_{co}^{xy}
    \approx 1- \prod_{v\in \A N_{xy}} [ 1- p_{co}^{xy}(v) ]
    \approx 1- [ 1- p_{co}^{xy}(\star) ]^{K_{xy}},
\end{align}
where the events corresponding to $1- p_{co}^{xy}(v)$, i.e., nodes not being cooperative with respect to $x$ and $y$, are regarded as independent of each other, as an approximation.

Thus $p_{co}$ is determined by averaging $p_{co}^{xy}$ over all $(x,y)$ pairs that are possible to create MCC problems.
It can be proven that these pairs are neighboring pairs $(x,y)$ satisfying ($d_i$ denoting the degree of a node $i$)
\begin{enumerate}
\item $d_x\ge 2$, $d_y\ge 2$, but not $d_x=d_y=2$, or
\item $d_x=d_y=2$, but $x$ and $y$ are not on the same three-cycle (triangle).
\end{enumerate}
This condition is satisfied by all neighboring pairs in a {\em connected} random network, because the connectivity requires a sufficiently high node degree ($5.18\log N$ where $N$ is the total number of nodes\cite{xue04conn}) which is much larger than 2. Therefore, taking expectation of \eqref{eq:pcoxy-approx} over all neighboring pairs using \lref{lem:epk} and \lref{lem:avgarea}-(c),
\begin{align}\label{eq:pcoxy}
p_{co} &= 1- \exp[ -p_{co}^{xy}(\star) \overline{K_{xy}}]\notag\\
  &\approx 1- \exp[ -1.84 n\; p_{co}^{xy}(\star) ].
\end{align}
This completes the analysis.

\subsection{Special Case: Single-Hop Networks}\label{sec:sghop}

Now that all nodes are in the communication range of each other, we have $p_{ni\text{-}oh}=p_{ni\text{-}cts}=1$ according to \pref{prop:notintfoh} and \ref{prop:notintfcts}, which leads to $p_{succ}=p_{oh}=p_{ctrl}$ according to \eqref{eq:poh_psucc}, and $p_{co}^{xy}(v)= p_{ctrl} \; p_{ctrl}^{\star}$ according to \eqref{eq:pcoxyv}. Hence \eqref{eq:pcoxy-approx} reduces to
\begin{align*}
p_{co}^{xy} = 1- (1- p_{ctrl} \; p_{ctrl}^{\star})^{K_{xy}},
\end{align*}
where $K_{xy}$ is the number of all possible cooperative nodes with respect to $x$ and $y$, leading to $K_{xy}=n-4$.
So, as the average of $p_{co}^{xy}$,
\begin{align}\label{eq:sghop}
    p_{co} = 1- (1- p_{ctrl}\; p_{ctrl}^{\star})^{n-4},
\end{align}
where $p_{ctrl}$ is given below, by solving the equations in \sref{sec:balance},
\begin{align*}
p_{ctrl} &= \frac{1}{2} (1 - \lambda T_d +
    \sqrt{1+\lambda T_d (\lambda T_d - 6)}\ ),
\end{align*}
and $p_{ctrl}^{\star}$ is given below, by reducing \eqref{eq:pctrlstar} with $w=0$,
\[ p_{ctrl}^{\star} = \frac{g(\lambda_w) - g(\lambda_c+\lambda_w)} {T_d - g(\lambda_c)} \]
\begin{align*}
\text{where \ }
\lambda_c &= \frac{1}{2} ( \frac{1- \sqrt{1+\lambda T_d (\lambda T_d - 6)}}
    {\lambda T_d^2} -\frac{3}{T_d}),\\
\lambda_w &= \frac{1- \sqrt{1+\lambda T_d (\lambda T_d - 6)}}{T_d} -\lambda.
\end{align*}

\section{Investigating $p_{co}$ with DISH}\label{sec:simu}

We verify the analysis in both single-hop and multi-hop networks and identify key findings therein. We also investigate the correlation between $p_{co}$ and network performance.

\subsection{Protocol Design and Simulation Setup}\label{sec:proto}

\subsubsection{Model-Based DISH}
This is a multi-channel MAC protocol based on the protocol framework described in \sref{sec:model}. Key part of its pseudo-code is listed below, where $S_{ctrl}$ is the control channel status (FREE/BUSY) detected by the node running the protocol, $S_{node}$ is the node's state (IDLE/TX/RX, etc.), $L_{queue}$ is the node's current queue length, and they are initialized as FREE, IDLE and 0, respectively.
The frame format of \texttt{McRTS} and \texttt{McCTS} is shown in \fref{fig:format}, where we can see that they carry channel usage information. A node that overhears \texttt{McRTS} or \texttt{McCTS} will cache the information in a {\em channel usage table} shown in \fref{fig:chtab}, where {\tt Until} is converted from \texttt{Duration} by adding the node's own clock.

\begin{figure}[htp]
\begin{minipage}[b]{0.59\linewidth}
\includegraphics[trim=1mm 2mm 2mm 1mm,clip,width=\linewidth]{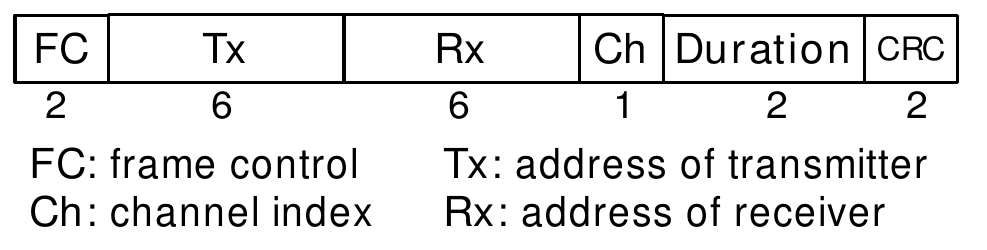}
\caption{Frame format of {\tt McRTS} and {\tt McCTS}.}\label{fig:format}
\end{minipage}\hfil
\begin{minipage}[b]{0.37\linewidth}
\rightline{
\includegraphics[width=\linewidth]{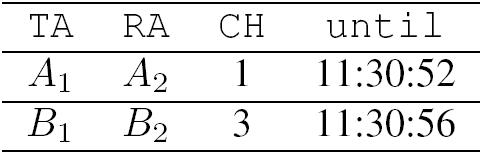}}
\caption{Channel usage table.}\label{fig:chtab}
\end{minipage}
\end{figure}

\floatname{algorithm}{Procedure}

\begin{algorithm}[htb]
\caption{PKT-ARRIVAL}[Called when a data packet arrives]
\label{alg:pktarr}
\begin{algorithmic}[1]
\STATE enqueue the packet, $L_{queue}$++
\IF{$S_{ctrl}=FREE \wedge S_{node}=IDLE \wedge L_{queue}=1$}
\STATE call ATTEMPT-RTS
\ENDIF
\end{algorithmic}
\end{algorithm}

\begin{algorithm}[htb]
\caption{ATTEMPT-RTS}[Called by PKT-ARRIVAL or CHECK-QUEUE]
\label{alg:rts}
\begin{algorithmic}[1]
\STATE construct a set $\A F$ of free channel indexes using channel usage table
\IF{$\A F\neq\phi$}
\STATE send \texttt{McRTS} with \texttt{CH}$:=$RANDOM$(\A F)$
\ELSE
\STATE Timer $\leftarrow \min(\texttt{until}-now)$
\WHILE{$S_{ctrl}=FREE\ \wedge$ Timer not expired}
\STATE wait \COMMENT{carrier sensing remains on}
\ENDWHILE
\IF{Timer expired}
\STATE call CHECK-QUEUE
\ELSE
\STATE call PASSIVE \COMMENT{receive a control message}
\ENDIF
\ENDIF
\end{algorithmic}
\end{algorithm}

\begin{algorithm}[htb]
\caption{CHECK-QUEUE}[Called when $S_{ctrl}=FREE \wedge S_{node}=IDLE$ changes from $FALSE$ to $TRUE$]
\label{alg:sendca}
\begin{algorithmic}[1]
\IF{$L_{queue}>0$}
\STATE Timer $\leftarrow$ RANDOM$(0,10b)$ \COMMENT{FAMA\cite{fama95,fama99}}
\WHILE{$S_{ctrl}=FREE\ \wedge$ Timer not expired}
\STATE wait \COMMENT{carrier sensing remains on}
\ENDWHILE
\IF{Timer expired}
\STATE call ATTEMPT-RTS
\ELSE
\STATE call PASSIVE \COMMENT{receive a control message}
\ENDIF
\ENDIF
\end{algorithmic}
\end{algorithm}

\ifdefined\EA
\begin{wrapfigure}{r}{0.4\textwidth}
  \begin{center}
    \includegraphics[trim=1mm 2mm 1mm 2mm,clip,width=0.35\textwidth]{fig/format_rtscts}
  \end{center}
  \caption{The frame format of \texttt{McRTS} and \texttt{McCTS}.}
  \label{fig:format}
\end{wrapfigure}
A node that overhears \texttt{McRTS} or \texttt{McCTS} (the frame format is shown in \fref{fig:format}) will cache the channel usage information carried by the packet in a {\em channel usage table}. Each entry of this table has four fields: \texttt{TA}, \texttt{RA}, \texttt{CH}, and \texttt{until}, where \texttt{until} is converted from \texttt{duration} (in the packet) by adding to the node's current time.
\fi

As is based on the system model, this protocol does not use a concrete DISH mechanism, i.e., cooperation is treated as a resource while not actually utilized.

\subsubsection{Ideal DISH}

This protocol is by adding an ideal cooperating mechanism to the model-based DISH. Each time when an MCC problem is created by nodes $x$ and $y$ and if at least one cooperative node is available, the node that is on the control channel, i.e., node $y$, will be informed without any message physically sent, and then back off to avoid the MCC problem.

\subsubsection{Real DISH}

In this protocol, cooperative nodes will physically send cooperative messages to inform a transmitter or receiver of the MCC problem so that it will backoff. We design this real DISH by adapting CAM-MAC\cite{tie06cam}.
The only change that we made is that, since in CAM-MAC a transmitter will send a PRA and a CFA, and a receiver will send a PRB and a CFB, during the control channel handshake, we change these control packet sizes such that $\|PRA\|+\|CFA\|=\|McRTS\|=\|PRB\|+\|CFB\|=\|McCTS\|$, where $\|\cdot\|$ gives the size of a packet.

\subsubsection{Simulation Setup}

There are six channels of data rate 1Mb/s each.
Data packets arrive at each node as a Poisson process.
The uniform traffic pattern as in the model is used.
Traffic load $\lambda$ (pkt/s), node density $n$ (1/$R^2$), and packet size $L$ (byte) will vary in simulations. In multi-hop networks, the network area is 1500m$\times$1500m and the transmission range is 250m. Each simulation is terminated when a total of 100,000 data packets are sent over the network, and each set of results is averaged over 15 randomly generated networks.

\subsection{Investigation with Model-Based DISH}\label{sec:verify}

\begin{figure}[tbp]
\centerline{
\subfigure[$p_{co}$ versus $\lambda$ ($n=8$), single-hop.]
    {\includegraphics[trim=4mm 1mm 11mm 7mm,clip,width=0.5\linewidth]{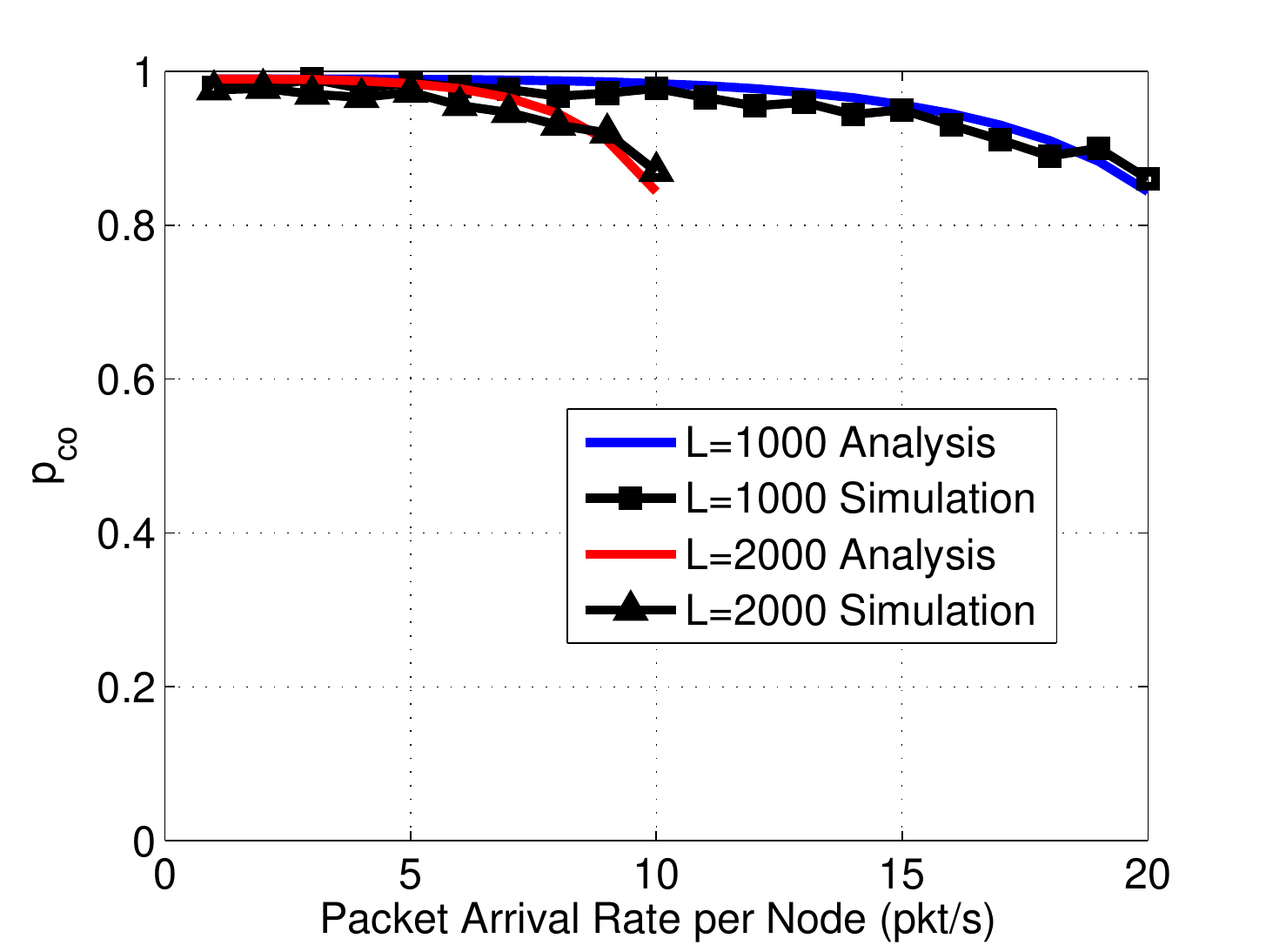}\label{fig:sg_load}}
\subfigure[$p_{co}$ versus $n$ ($\lambda=10$), single-hop.]
    {\includegraphics[trim=4mm 1mm 11mm 7mm,clip,width=0.5\linewidth]{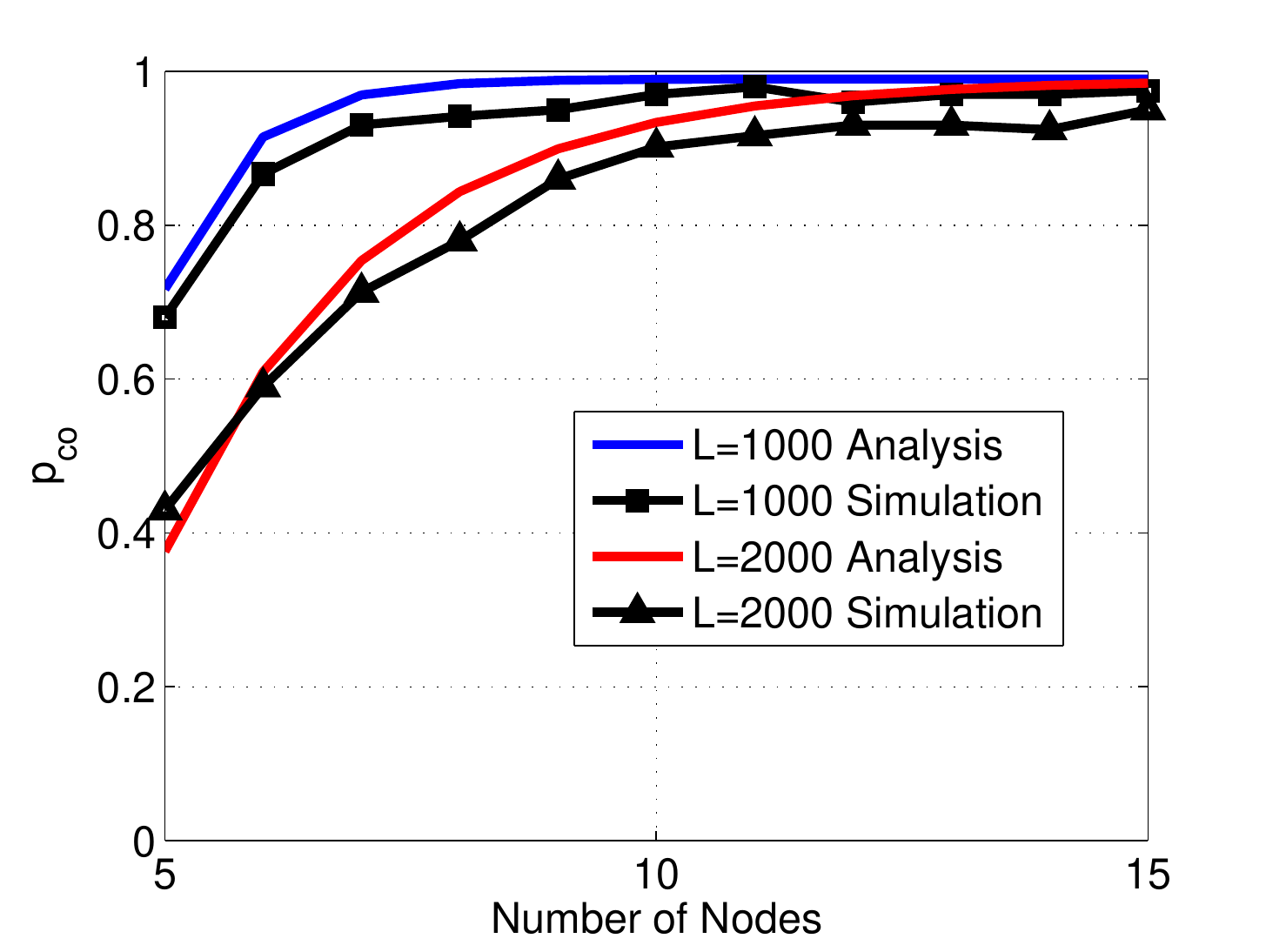}\label{fig:sg_node}}}
\centerline{
\subfigure[$p_{co}$ versus $\lambda$ ($n=10$), multi-hop.]
    {\includegraphics[trim=4mm 1mm 11mm 7mm,clip,width=0.5\linewidth]{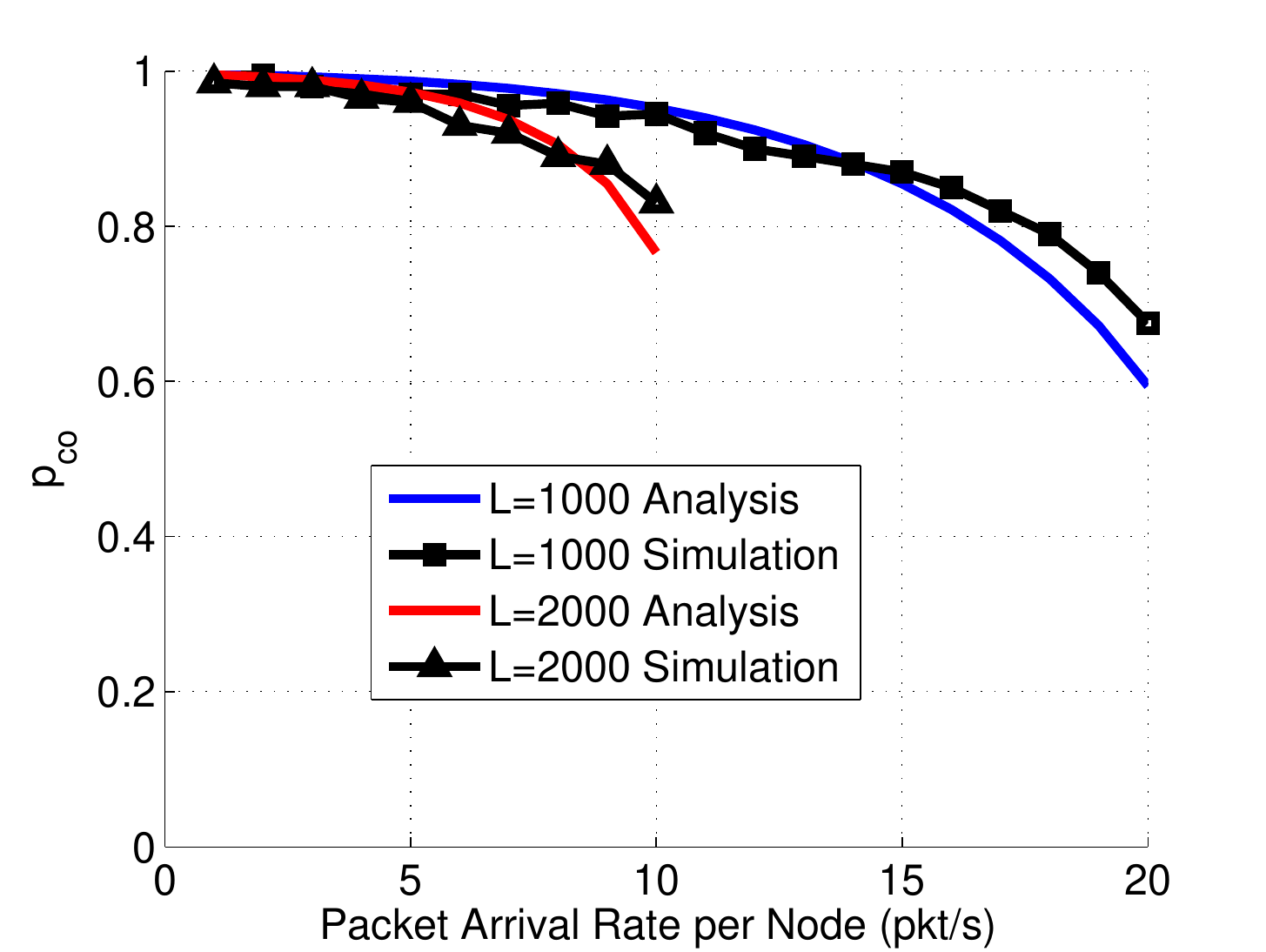}\label{fig:mt_load}}
\subfigure[$p_{co}$ versus $n$ ($\lambda=10$), multi-hop.]
    {\includegraphics[trim=4mm 1mm 11mm 7mm,clip,width=0.5\linewidth]{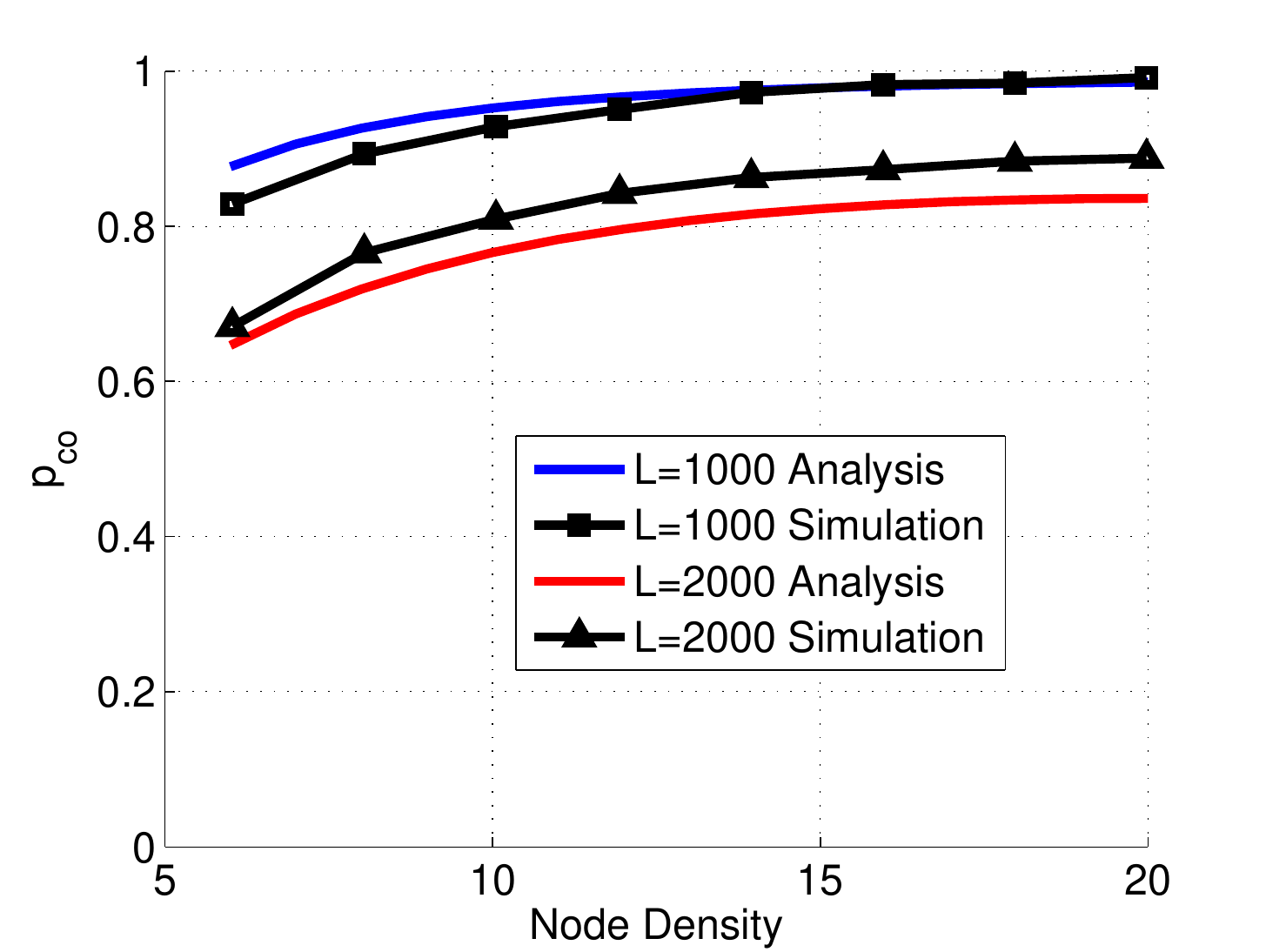}\label{fig:mt_node}}}
\caption{Impact of traffic load and node density, with different packet sizes. The value ranges of X axes are chosen such that the network is stable.}
\label{fig:result}
\end{figure}

The $p_{co}$ obtained via analysis and simulations are compared in \fref{fig:result}. We see a close match between them, with a deviation of less than 5\% in almost all single-hop scenarios, 
and less than 10\% in almost all multi-hop scenarios. Particularly, the availability of cooperation is observed to be at a high level ($p_{co}>0.7$ in most cases), which suggests that a large percentage of MCC problems would be avoided by exploiting DISH, and DISH is feasible to use in multi-channel MAC protocols. ({\bf Finding 1})

Specifically, \fref{fig:sg_load} and \fref{fig:mt_load} consistently show that, in both single-hop and multi-hop networks, $p_{co}$ monotonically decreases as $\lambda$ increases. The reasons are two folds. First, as traffic grows, each node spends more time on data channels for data transmission and reception, which reduces $p_{ctrl}$ and hence the chance of overhearing control messages ($p_{oh}$), resulting in lower $p_{co}$. Second, as the control channel is the rendezvous to set up all communications, larger traffic intensifies the contention and introduces more interference to the control channel, which is hostile to messages overhearing and thus also reduces $p_{co}$.

\fref{fig:sg_node} and \fref{fig:mt_node} show that $p_{co}$ monotonically increases as $n$ increases, and is concave. The increase of $p_{co}$ is because MCC problems are more likely to have cooperative nodes under a larger node population, while the deceleration of the increase is because more nodes also generate more interference to the control channel.

An important message conveyed by this observation is that, although a larger node density creates more MCC problems (e.g., more channel conflicts as data channels are more likely to be busy), it also boosts the availability of cooperation which avoids more MCC problems. This implies that the performance degradation can be mitigated. ({\bf Finding 2})

In both single-hop and multi-hop networks, a larger packet size $L$ corresponds to a lower $p_{co}$. However, note that this is observed under the same {\em packet} arrival rate (pkt/s), which means actually a larger {\em bit} arrival rate for a larger $L$, and can be explained by the previous scenarios of $p_{co}$ versus $\lambda$. Now if we consider the same {\em bit} arrival rate, by examining the two analysis curves in \fref{fig:mt_load} where we compare $p_{co}$ with respect to the same $\lambda\cdot L$ product, e.g., ($\lambda=5,L=2000$) versus ($\lambda=10,L=1000$), and ($\lambda=10,L=2000$) versus ($\lambda=20,L=1000$), then we will see that a larger $L$ corresponds to a {\em higher} $p_{co}$, which is contrary to the observation under the same {\em packet} arrival rate. The explanation is that, for a given bit arrival rate, increasing $L$ reduces the number of packets and hence {\em fewer control channel handshakes} are required, thereby alleviating control channel interference. ({\bf Finding 3})

\begin{figure}[tbp]
\centering
\includegraphics[trim=1cm 7mm 1cm 13mm,clip,width=0.7\linewidth]{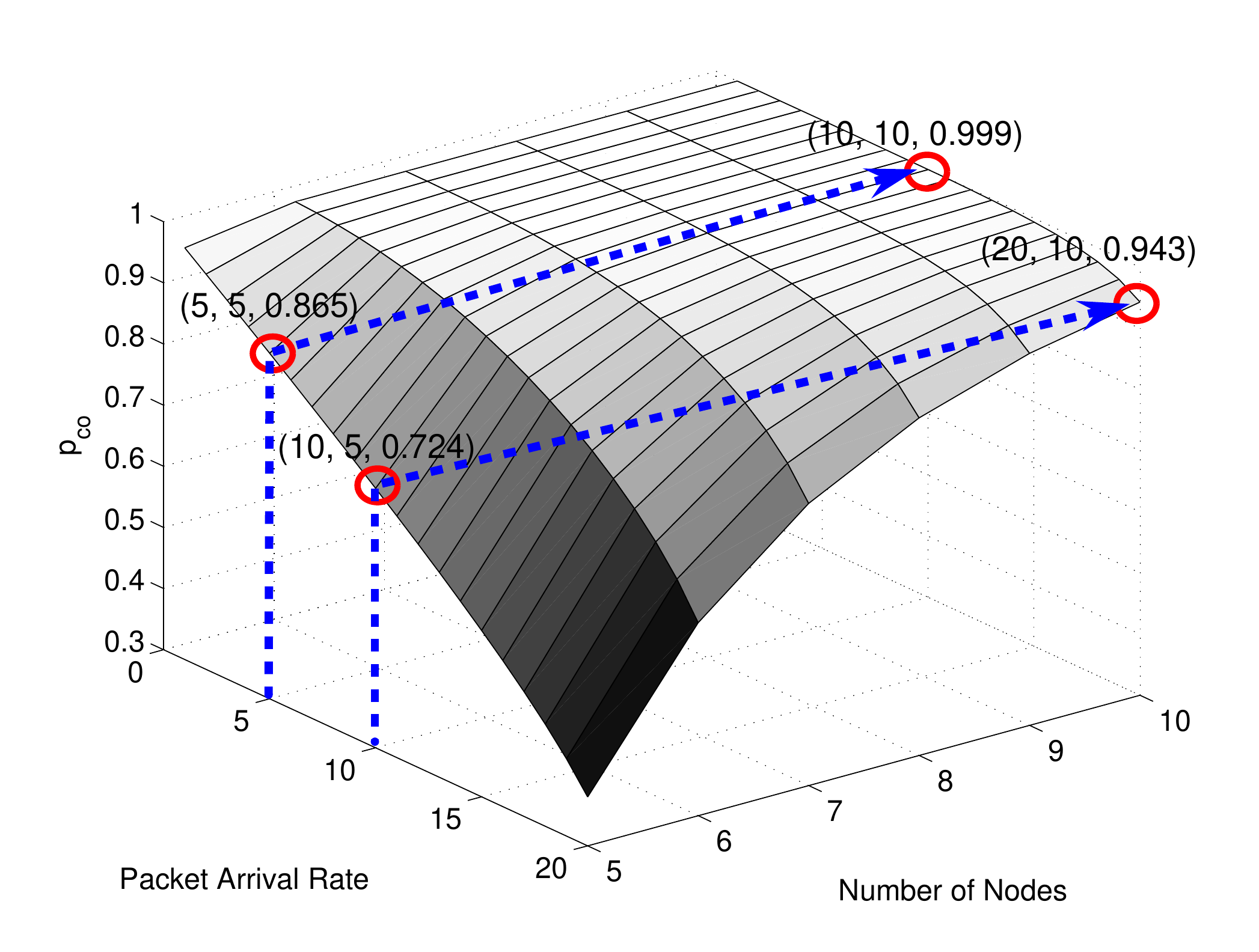}
\caption{$p_{co}$ versus $\lambda$ and $n$. Each of the two arrows indicates a multiplicative increase of $\lambda$ and $n$ with the same factor (two).}
\label{fig:3d_sg_L1k}
\end{figure}

The above results indicate that node density and traffic load affect the availability of cooperation in {\em opposite} ways. This section aims to find which one dominates over the other.
In \fref{fig:3d_sg_L1k}, we plot the relationship of $p_{co}$ versus $\lambda$ and $n$, given $L=1000$ and based on the {\em analytical} result for single-hop networks.
We multiplicatively increase $\lambda$ and $n$ with the same factor (two), and find that, when increasing ($\lambda,n$) from (5,5) to (10,10), $p_{co}$ keeps {\em increasing} from 0.865 to 0.999, and when increasing ($\lambda$,n) from (10,5) to (20,10), $p_{co}$ keeps {\em increasing} from 0.724 to 0.943.

This investigation shows that $n$ is the dominating factor over $\lambda$ that determines the variation of $p_{co}$. This implies that DISH networks should have better scalability than non-DISH networks, since $p_{co}$ increases when both traffic load and node density scale up. ({\bf Finding 4})

\subsection{Investigation with Ideal DISH}\label{sec:ideal}

The results of comparison are shown in \fref{fig:stable}, where $p_{co}$ with ideal DISH well matches $p_{co}$ of analysis. This confirms {\bf Findings 1-3}, and we speculate the reasons to be as follows. With ideal DISH, a transmitter may be informed of a deaf terminal problem and thus will backoff for a fairly long time, which leads to {\em fewer} {\tt McRTS} being sent. On the other hand, a node may also be informed of a channel conflict problem and thus will re-select channel and retry shortly, which leads to {\em more} control messages being sent. Empirically, the latter case has more significant effect, which means that, overall, there will be an {\em increase} of control messages being sent. This boosts interference and thus would reduce $p_{co}$. However, nodes will also stay {\em longer} on the control channel due to less use of conflicting data channels, which would elevate $p_{co}$. Consequently, $p_{co}$ does not change noticeably.

\begin{figure}[tbp]
\centerline{
\subfigure[Impact of $\lambda$ ($n=10$).]
    {\includegraphics[trim=11mm 1mm 1cm 7mm,clip,width=0.5\linewidth]{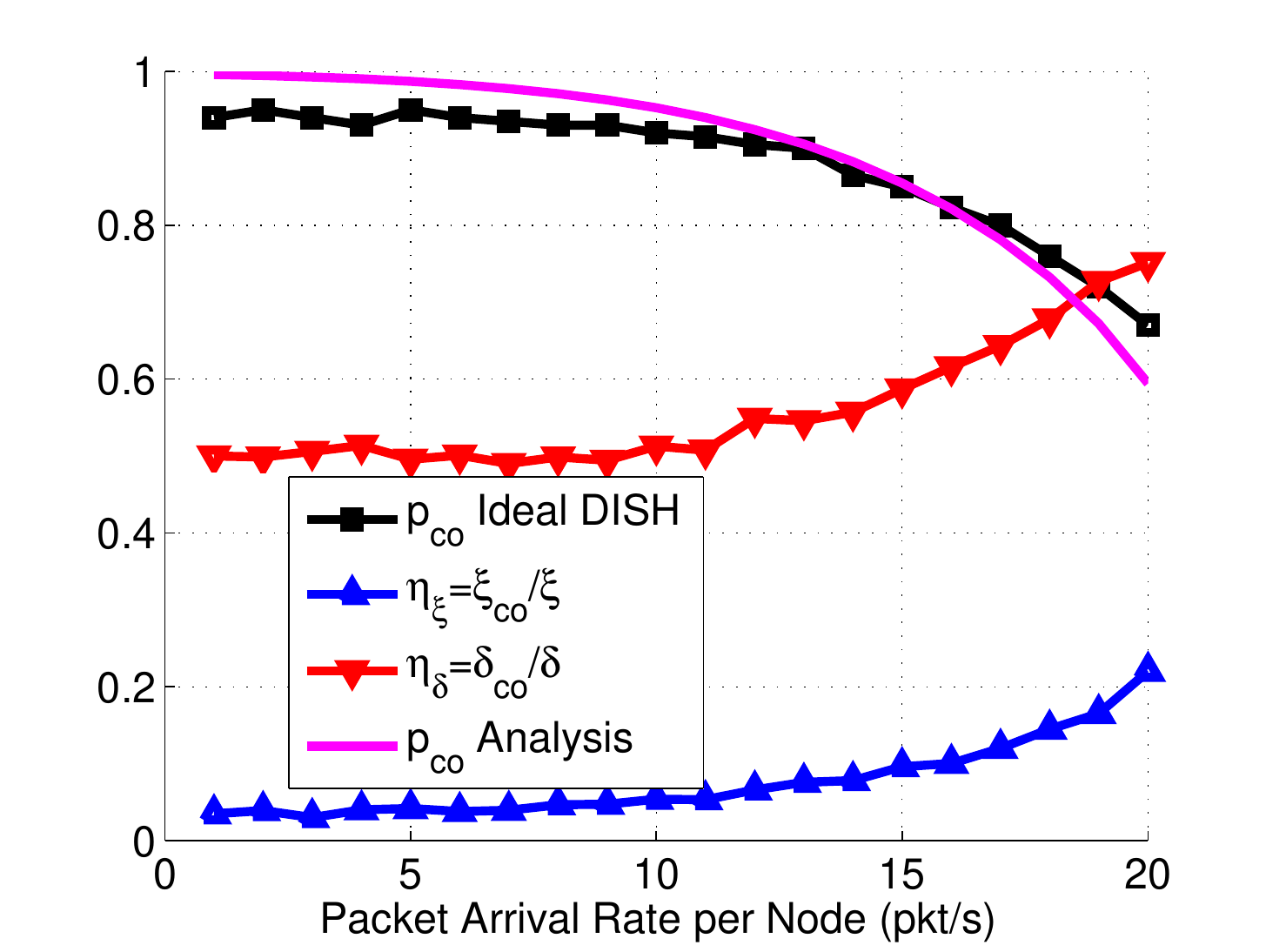}\label{fig:stable_ratio_load}}
\subfigure[Impact of $n$ ($\lambda=10$).]
    {\includegraphics[trim=11mm 1mm 1cm 7mm,clip,width=0.5\linewidth]{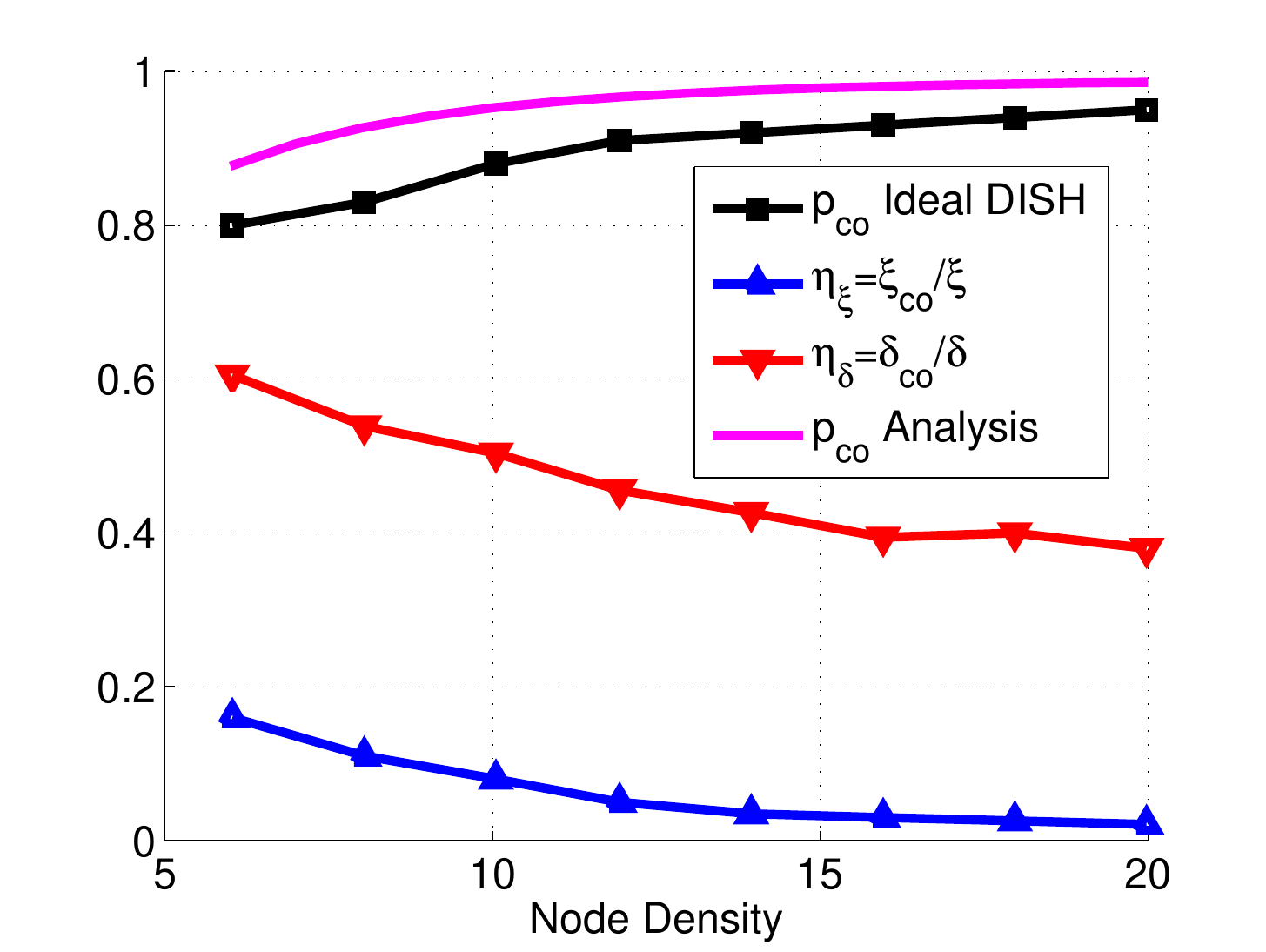}\label{fig:stable_ratio_node}}}
\caption{Investigating $p_{co}$ with ideal DISH in stable networks. This includes (i) verification of analysis, and (ii) correlation between $p_{co}$ and ($\eta_\xi, \eta_\delta$) (ratio of data collision, ratio of packet delay). $L=1000$ byte.}
\label{fig:stable}
\vspace{-3mm}\end{figure}

We investigate how $p_{co}$ correlates to network performance --- specifically, data channel collision rate $\xi$, packet delay $\delta$, and aggregate throughput $S$.
We consider both stable networks and saturated networks under multi-hop scenarios.

\ifdefined\JNL
\begin{figure*}[tbp]
\centerline{
\subfigure[Data channel collision rate.]
    {\includegraphics[trim=3mm 0 1cm 7mm,clip,width=0.31\linewidth]{stable_coll_node}\label{fig:stable_coll_node}}
\subfigure[Packet delay.]
    {\includegraphics[trim=3mm 0 1cm 7mm,clip,width=0.31\linewidth]{stable_delay_node}\label{fig:stable_delay_node}}
\subfigure[Evaluation of $p_{co}$, $\eta_\xi$ and $\eta_\delta$.]
    {\includegraphics[trim=6mm 0 1cm 7mm,clip,width=0.31\linewidth]{stable_ratio_node}\label{fig:stable_ratio_node}}}
\caption{Performance evaluation in stable networks by varying node density ($\lambda=10$ pkt/s).}
\label{fig:stable_node}
\end{figure*}
\begin{figure*}[tbp]
\centerline{
\subfigure[Data channel collision rate.]
    {\includegraphics[trim=3mm 0 1cm 7mm,clip,width=0.31\linewidth]{stable_coll_load}\label{fig:stable_coll_load}}
\subfigure[Packet delay.]
    {\includegraphics[trim=3mm 0 1cm 7mm,clip,width=0.31\linewidth]{stable_delay_load}\label{fig:stable_delay_load}}
\subfigure[Evaluation of $p_{co}$, $\eta_\xi$ and $\eta_\delta$.]
    {\includegraphics[trim=6mm 0 1cm 7mm,clip,width=0.31\linewidth]{stable_ratio_load}\label{fig:stable_ratio_load}}}
\caption{Performance evaluation in stable networks by varying traffic load ($n=10/r^2$).}
\label{fig:stable_load}
\vspace{-5mm}\end{figure*}
\fi

In stable networks, we measure ($\xi,\delta$) and $(\xi_{co},\delta_{co})$ when without and with cooperation (ideal DISH), respectively. Then we compute $\eta_\xi=\xi_{co}/\xi$ and $\eta_\delta=\delta_{co}/\delta$ to compare to $p_{co}$ with ideal DISH.  The first set of results, by varying traffic load $\lambda$, is shown in \fref{fig:stable_ratio_load}. %
\ifdefined\JNL
Interestingly, although as seen from \fref{fig:stable_coll_load} and \fref{fig:stable_delay_load}, cooperation reduces data collisions and packet delay similarly as the previous results, the ratio $\eta_\xi$ and $\eta_\delta$ increases, as seen from \fref{fig:stable_ratio_load} and is contrary to the previous results in \fref{fig:stable_ratio_node}. This actually confirms one of our previous findings in \sref{sec:verify}: increasing $\lambda$ suppresses $p_{co}$, and thus the performance improvement via cooperation is diminishing.
\fi
We observe that the two {\em ascending} and {\em convex} curves of $\eta_\xi$ and $\eta_\delta$ approximately {\em reflect} the {\em descending} and {\em concave} curve of $p_{co}$, which hints at a {\em linear} or {\em near-linear} relationship between $p_{co}$ and these two performance ratios. That is, $\eta_\xi+p_{co}\approx c_1,\ \eta_\delta+p_{co}\approx c_2$, where $c_1$ and $c_2$ are two constants. The second set of results, by varying node density $n$, is shown in \fref{fig:stable_ratio_node}.
\ifdefined\JNL
As seen from \fref{fig:stable_coll_node}, the data channel collision rate without cooperation, $\xi$, increases exponentially, whereas the rate with cooperation, $\xi_{co}$, increases very mildly and is extremely low. Likewise, \fref{fig:stable_delay_node} indicates that the packet delay with and without cooperation increase in a slow and a fast speed, respectively. Finally, the two ratios, $\eta_\xi$ and $\eta_\delta$, are depicted in \fref{fig:stable_ratio_node} together with $p_{co}$ of analysis and simulation,
\fi
On the one hand, $\eta_\xi$ and $\eta_\delta$ decreases as $n$ increases, which is contrary to \fref{fig:stable_ratio_load}. This confirms our earlier observations: $n$ is amicable whereas $\lambda$ is hostile to $p_{co}$ (the smaller $\eta_\xi$ and $\eta_\delta$, the better performance cooperation offers). On the other hand, the correlation between $p_{co}$ and the performance ratios is found again: as $p_{co}$ increases on a concave curve, it is reflected by $\eta_\xi$ and $\eta_\delta$ which decrease on two
convex curves.

In saturated networks, we vary node density $n$ and measure aggregate throughput without and with cooperation (ideal DISH), as $S$ and $S_{co}$, respectively. Then we compute $\eta_S=S/S_{co}$ (note that this definition is inverse to $\eta_\xi$ and $\eta_\delta$, such that $\eta_S\in[0,1]$) to compare to $p_{co}$ with ideal DISH. The results are summarized in \fref{fig:satu_ratio}.
We see that (i) $p_{co}$ grows with $n$, which conforms to Finding 2, and particularly, (ii) the declining and convex curve of $\eta_S$ reflects the rising and concave curve of $p_{co}$, which is consistent with the observation in stable networks. In addition, here $p_{co}$ is lower than the $p_{co}$ in stable networks. This is explained by our earlier result that higher traffic load suppresses $p_{co}$.

\ifdefined\JNL
$S$ levels off early (around $n=10/r^2$) and starts to decline later (around $n=16/r^2$), but $S_{co}$ keeps increasing, with the slope only slightly decreasing, and nearly doubles $S$ at $n=20$.
\fi

In summary, the experiments in stable networks and saturated networks both demonstrate a strong correlation ({\em linear} or {\em near-linear} mapping) between $p_{co}$ and network performance ratio in terms of typical performance metrics. This may significantly simplify performance analysis for cooperative networks via bridging the {\em nonlinear} gap between network parameters and $p_{co}$, and also suggests that $p_{co}$ be used as an appropriate performance indicator itself. ({\bf Finding 5})

\begin{figure}[tbp]
\centering
\includegraphics[trim=4mm 1mm 6mm 6mm,clip,width=0.55\linewidth]{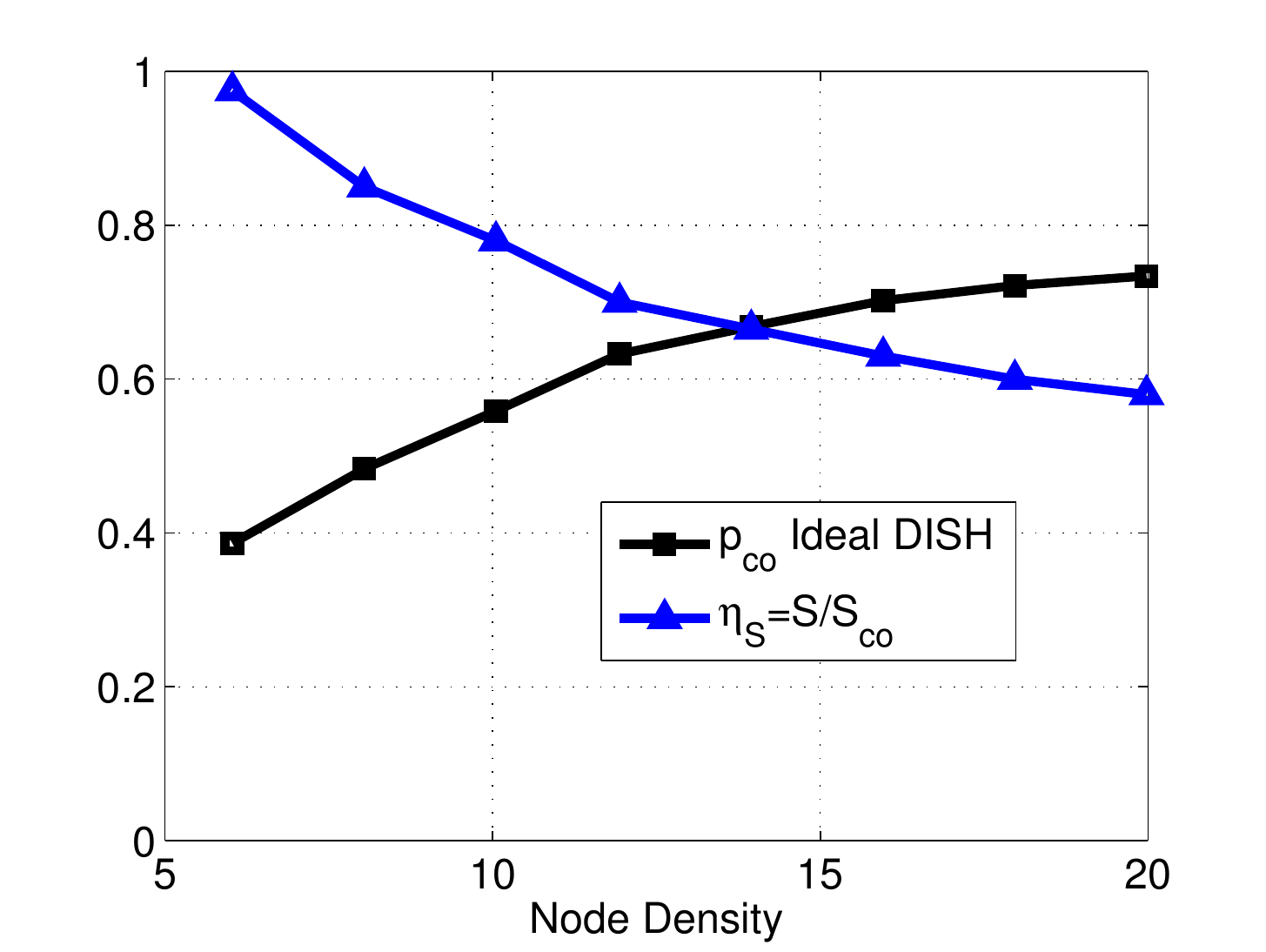}
\caption{Investigating $p_{co}$ with ideal DISH in saturated networks: correlation between $p_{co}$ and $\eta_S$ (throughput ratio). $L=1000$ byte.}
\label{fig:satu_ratio}
\vspace{-3mm}\end{figure}

The explanation to this linear or near-linear relationship should involve intricate network dynamics. We speculate that the rationale might be that (i) MCC problems are an essential performance {\em bottleneck} to multi-channel MAC performance, and (ii) $p_{co}$ is equivalent to the {\em ratio} of MCC problems that can be avoided by DISH. In any case, we reckon that this observation may spur further studies and lead to more thought-provoking results.

\subsection{Investigation with Real DISH}\label{sec:cammac}

Due to space constraints, we discuss the results without presenting figures because the results were found to be similar to those with ideal DISH.

The deviation between simulations and analysis, which maximally reached $\sim$15\%, was found to be slightly larger than that with model-based or ideal DISH. However, the overall trends still match, and {\bf Findings 1-4} are confirmed.  For explanation, we speculate that, in real DISH, there are much more control messages (including cooperative messages) being sent, which result in more control channel interference and thus would diminish $p_{co}$. But on the other hand, nodes stay longer on the control channel due to the same reason as for the case of ideal DISH (i.e., less use of conflicting data channels), which would increase $p_{co}$. As a result, $p_{co}$ with real DISH does not deviate significantly from the analysis.

The near-linear relationship between $p_{co}$ and ($\eta_\xi, \eta_\delta, \eta_S$) was observed. This confirms {\bf Finding 5}. As a possible explanation, the remark following Finding 5 in \sref{sec:ideal} applies.  Although the absolute values of these quantities were found to differ from those with ideal DISH (ranging between 3-18\%), this is not the primary concern.

\section{Conclusion}\label{sec:conc}

Distributed Information SHaring (DISH) represents a mechanism different from data relaying to exploit cooperative diversity. This paper gives the first theoretical treatment of this new notion of cooperation, by addressing the availability of cooperation via a metric $p_{co}$. Instead of directly analyzing throughput which is an open problem in general and rendered much more complicated when considered together with multiple channels and DISH, our approach is to analyze $p_{co}$ first and then correlate $p_{co}$ with performance metrics including throughput. We conduct analysis in a multi-hop multi-channel wireless network and, to verify its validity and study its implications, investigate $p_{co}$ with three different contexts of DISH: model-based DISH, ideal DISH, and real DISH. The investigation validates that our analysis accurately captures the interaction among network parameters, which allows us to draw important findings of $p_{co}$ with respect to network dynamics. It also reveals a near-linear relationship between $p_{co}$ and network performance, which may greatly aid in performance analysis for cooperative networks and also suggests $p_{co}$ to be a proper performance indicator itself.

This work is the first that explicitly presents DISH, together with a detailed study offering meaningful insights into understanding cooperation. Based on our findings, we conclude that $p_{co}$ is a useful metric capable of characterizing the performance of {\em DISH networks}. We contend that DISH is useful and practical enough to be a part of future cooperative communication networks.

\bibliographystyle{IEEEtran}


\appendix

\section{Proofs and Derivations}

\subsection{Proof of \pref{prop:notintfoh}}
\begin{proof}
In the case of $u \in \A N_{vi}$, no matter $u$ is on the control channel at $s_i$, or is on a data channel at $s_i$ but switches to the control channel before $s_i+b$, it will sense a busy control channel (due to CSMA) and thus keep silent.

\begin{figure}[htbp]
\centering\includegraphics[width=0.88\linewidth]{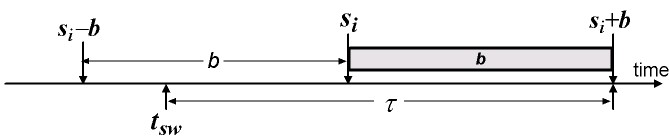}
\caption{The vulnerable period of $v$ is $[s_i-b,s_i+b]$, in which node $u\in \A N_{v\backslash i}$ should not start transmission on the control channel.}
\label{fig:vp}
\end{figure}

In the case of $u \in \A N_{v\backslash i}$, see \fref{fig:vp}. Note that the vulnerable period of $v$ is $[s_i-b, s_i+b]$ instead of $[s_i, s_i+b]$, because a transmission started within $[s_i-b,s_i]$ will end within $[s_i, s_i+b]$. Therefore, by the total probability theorem,
\begin{align*}
p_{ni\text{-}oh}
 =& \Pr[\A C_u(s_i-b)] \cdot \Pr[\A I_u(s_i-b, s_i+b)|\A C_u(s_i-b)]\\
  &+\Pr[\overline{\A C_u(s_i-b)}] \cdot \Pr[\A I_u(s_i-b, s_i+b) |\overline{\A C_u(s_i-b)}]\\
 =&\ p_{ctrl} \cdot e^{-2 \lambda_c b} +(1-p_{ctrl})
  \cdot (1-\frac{2b}{T_d} + \frac{1- e^{-2 \lambda_c b}}{\lambda_c T_d}),
\end{align*}
where $\Pr[\A I_u(s_i-b, s_i+b) |\overline{\A C_u(s_i-b)}]$ is solved by \pref{prop:switch}.
\end{proof}

\subsection{Proof of \pref{prop:notintfcts}}
\begin{proof}
The case of $u \in \A N_{ij}$ follows the same line as the proof for \pref{prop:notintfoh}.   In the case of $u \in \A N_{i\backslash j}$, the only difference from \pref{prop:notintfoh} is that now we are implicitly given the fact that $i$ was transmitting \texttt{McRTS} during $[s_j-b, s_j]$. This excludes $i$'s any neighbor $u$ interfering in $[s_j-b, s_j]$. Therefore $i$'s vulnerable period is $[s_j, s_j+b]$ instead of $[s_j-b, s_j+b]$ as compared to \pref{prop:notintfoh}. So
\begin{align*}
p_{ni\text{-}cts} =& \Pr[\A C_u(s_j-b)] \cdot \Pr[\A I_u(s_j, s_j+b)|\A C_u(s_j-b)]\\
 &+\Pr[\overline{\A C_u(s_j-b)}]\cdot \Pr[\A I_u(s_j, s_j+b)|\overline{\A C_u(s_j-b)}].
\end{align*}
Note that we condition on $\A C_u(s_j-b)$ instead of $\A C_u(s_j)$, because $s_j$ is not an {\em arbitrary} time due to $i$'s \texttt{McRTS} transmission during $[s_j-b, s_j]$, which leads to $\Pr[\A C_u(s_j)]\neq p_{ctrl}$.

First, $\Pr[\A I_u(s_j, s_j+b)|\A C_u(s_j-b)]=1$. This is because, as $\A C_u(s_j-b)\Leftrightarrow \A S_u(s_j-b, s_j)$ which is easy to show, $u$ will successfully overhear $i$'s \texttt{McRTS}, and hence will keep silent in the next period of $b$ to avoid interfering with $i$ receiving \texttt{McCTS}.

Next consider $\Pr[\A I_u(s_j, s_j+b)|\overline{\A C_u(s_j-b)}]$ where $u$ is on a data channel at $s_j-b$. If $u$ switches to the control channel (i) before $s_j$, it will be suppressed by $i$'s \texttt{McRTS} transmission until $s_j$, and thus the vulnerable period of $i$ receiving \texttt{McCTS} is $[s_j, s_j+b]$, (ii) within $[s_j, s_j+b]$, this has been solved by \pref{prop:switch}, or (iii) after $s_j+b$, the probability to solve is obviously 1. Therefore,
\begin{align*}
\Pr[&\A I_u(s_j, s_j+b)|\overline{\A C_u(s_j-b)}]
= \Pr[\Omega_u(s_j-b, s_j)]\; e^{-\lambda_c b}\\
    &+ \Pr[\Omega_u(s_j, s_j+b)]\,(1-\frac{b}{T_d} + \frac{1- e^{-\lambda_c b}}{\lambda_c T_d})\\
    &+ \{1-\Pr[\Omega_u(s_j-b, s_j)]-\Pr[\Omega_u(s_j, s_j+b)]\}\times 1.
\end{align*}
According to \eqref{eq:switch}, $\Pr[\Omega_u(s_j-b, s_j)]=\Pr[\Omega_u(s_j, s_j+b)]=b/T_d$. Then by substitution the proposition is proven.
\end{proof}

\subsection{Derivation of Equation~\eqref{eq:poh-vi}}
\begin{proof}
Based on the proof for the case $u \in \A N_{vi}$ in \pref{prop:notintfoh}, it is easy to show that $\A S_v(s_i, s_i+b)\Leftrightarrow\A C_v(s_i)$. Hence
\[ \Pr[\A S_v(s_i, s_i+b)] = \Pr[\A C_v(s_i)] = p_{ctrl}. \]
Treating events $\A S_v(s_i, s_i+b)$ (node $v$ is silent on the control channel) and $\A I_u(s_i, s_i+b)$ (node $u$ does not interfere the control channel) being independent of each other, as an approximation,
we have
\[ \Pr[\A O (v\leftarrow i)] \approx p_{ctrl} \prod_{u \in \A N_{v\backslash i}} p_{ni\text{-}oh}
        = p_{ctrl} \; p_{ni\text{-}oh}^{K_{v\backslash i}}. \]
\end{proof}

\subsection{Derivation of Equation~\eqref{eq:poh_psucc}}
\begin{proof}
Taking the expectation of $\Pr[\A O (v\leftarrow i)]$ (given by \eqref{eq:poh-vi}) over all neighboring $(v,i)$ pairs using \lref{lem:epk} and \lref{lem:avgarea}-(a):
\begin{align*}
p_{oh} \approx p_{ctrl}\; \bb{E}[p_{ni\text{-}oh}^{K_{v\backslash i}}]
    \approx p_{ctrl}\; \exp[-1.30 n (1-p_{ni\text{-}oh})].
\end{align*}

To solve for $p_{succ}$, notice that for a control channel handshake to be successful, (i) the \texttt{McRTS} must be successfully received by the receiver, with probability $p_{oh}$, and (ii) the \texttt{McCTS} must be successfully received by the transmitter, with probability $\bb E[p_{ni\text{-}cts}^{K_{i\backslash j}}]$ based on \pref{prop:notintfcts} (assuming that $p_{ni\text{-}cts}$ holds for nodes in $\A N_{i\backslash j}$ independently, as an approximation). Therefore,
\begin{align*}
p_{succ} \approx p_{oh} \mathbb{E}[p_{ni\text{-}cts}^{K_{i\backslash j}}]
    \approx p_{oh} \exp[-1.30n(1-p_{ni\text{-}cts})].
\end{align*}
\end{proof}

\subsection{Derivation of Equation~\eqref{eq:pctrlstar}}
\begin{proof}
Recall that node $v$ must stay continuously on the control channel during $[s_x,s_y]$. Let $\tau_c=s_y-s_x$ and suppose $v$ switches to a data channel at $s_x+\tau_w$, then we need $\tau_w>\tau_c$. Hence $p_{ctrl}^{\star} = \Pr(\tau_w>\tau_c)$, where $\tau_c\in [0,T_d]$.

\begin{figure}[tbp]
\centering
\includegraphics[width=0.98\linewidth]{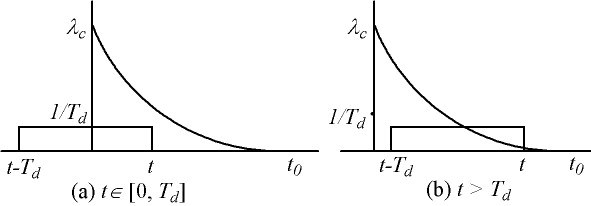}
\caption{The convolution of $\frac{1}{T_d}$ ($t\in [0,T_d]$) and $\lambda_c e^{-\lambda_c t}$ ($t>0$).}
\label{fig:convolute}
\vspace{-5mm}\end{figure}

Denote by $f_{\tau_c}(t)$ the pdf of an unbounded $\tau_c$ ($s_y\in (s_x,\infty)$). The fact that a MCC problem is created by $x$ and $y$ (at $s_y$) implies that $y$ missed $x$'s control message (at $s_x$). This is due to one of the following: (i) $y$ is on the control channel at $s_x$ but interfered, in which case $f_{\tau_c}(t)$ is $\lambda_c e^{-\lambda_c t}$ (ignoring the short interference period which is in the magnitude of $b$, while $\tau_c$ is in the magnitude of $T_d$), (ii) $y$ is on a data channel at $s_x$, in which case $y$ must switch to the control channel before $s_y$. Again see \fref{fig:switch}, where $t_1$ and $t_2$ are now $s_x$ and $s_y$, respectively. Let $\tau_1=t_{sw}-s_x$ and $\tau_2=s_y-t_{sw}$, then $\tau_c=\tau_1+\tau_2$. Note that $\tau_1$ is uniformly distributed in $[0,T_d]$, $\tau_2$ is exponentially distributed with the mean of $1/\lambda_c$, and $\tau_1$ and $\tau_2$ can be regarded as independent. Therefore, $f_{\tau_c}(t)$ is the convolution of $\frac{1}{T_d}$ ($t\in [0,T_d]$) and $\lambda_c e^{-\lambda_c t}$ ($t>0$), which can be calculated by referring to \fref{fig:convolute}, to be $f_{\tau_c}^d (t) =$
\begin{align*}
\frac{1- e^{-\lambda_c t} }{T_d} [u(t)-u(t-T_d)]
         + \frac{ e^{-\lambda_c t} }{T_d} (e^{\lambda_c T_d} -1) u(t-T_d).
\end{align*}
where $u(\cdot)$ is the unit step function.

A weighted sum of the above cases (i) and (ii) gives
\begin{align*}
  f_{\tau_c}(t) = w\; \lambda_c e^{-\lambda_c t} + (1-w)\; f_{\tau_c}^d (t)
\end{align*}
where $w$ is the weight for case (i). To determine $w$, note that the probability of case (ii) is $1-p_{ctrl}$, and the probability of case (i) is $p_{ctrl} (1-p_{ni\text{-}oh}^{K_{y\backslash x}})$ (using \eqref{eq:poh-vi}) whose mean is $ p_{ctrl} (1-\exp[-1.30 n(1-p_{ni\text{-}oh})])$. Therefore
\begin{align*}
w =\frac{p_{ctrl} [1-e^{-1.30n(1-p_{ni\text{-}oh})}]}
{p_{ctrl}[1-e^{-1.30n(1-p_{ni\text{-}oh})}] + (1-p_{ctrl})}
  = \frac{p_{ctrl}-p_{oh}}{1-p_{oh}}.
\end{align*}

Finally we compute $p_{ctrl}^{\star} = \Pr(\tau_w>\tau_c)$ using $f_{\tau_c}(t)$. Recall that $f_{\tau_c}(t)$ is the pdf of an unbounded $\tau_c$ but $\tau_c$ is in fact bounded within $[0,T_d]$, therefore its actual pdf is\\
$f_{\tau_c}(t) / \int_0^{T_d} f_{\tau_c}(t) dt$. Assuming that $\tau_w$ is exponentially distributed with mean $1/\lambda_w$, we have
\begin{align*}
p_{ctrl}^{\star} = \bb E_{\tau_c\in[0,T_d]}\;\Pr(\tau_w>\tau_c)
    = \int_0^{T_d} e^{-\lambda_w t} \frac{f_{\tau_c}(t)}{\int_0^{T_d} f_{\tau_c}(t) dt} dt
\end{align*}
which reduces to \eqref{eq:pctrlstar}.  For $\lambda_w$, noticing that it is the average rate of a node on the control channel switching to data channels, which happens when a node successfully initiates a control channel handshake via \texttt{McRTS} or sends a \texttt{McCTS}, we have $\lambda_w=\lambda_{rts} p_{succ}+\lambda_{cts}$.
\end{proof}

\end{document}